\documentclass[a4paper,11pt,DIV14]{scrartcl}
\pdfoutput=1

\usepackage[utf8x]{inputenc}
\usepackage{amsmath}
\usepackage{bbm}
\usepackage{amsthm}
\usepackage{amssymb}
\usepackage{hyperref}

  \newtheorem{theorem}{Theorem}
  \newtheorem{lemma}{Lemma}
  \newtheorem{corollary}{Corollary}
  
  \theoremstyle{definition}
  
  \newtheorem{definition}{Definition}

\usepackage[affil-it]{authblk}
\usepackage{algorithmicx}
\usepackage{algorithm}
\usepackage[noend]{algpseudocode}

\usepackage{url}
\usepackage{tikz}
\usetikzlibrary{calc}
\usetikzlibrary{patterns}
\usetikzlibrary[patterns]
\usepgflibrary{patterns}
\usepgflibrary[patterns]

\title{Sublinear Estimation of Weighted Matchings in Dynamic Data Streams\thanks{Supported by Deutsche Forschungsgemeinschaft, grant BO 2755/1-2 and within the Collaborative
Research Center SFB 876, project A2}}
\author{Marc Bury, Chris Schwiegelshohn}
\affil{Efficient Algorithms and Complexity Theory, %Computer Science~II, 
TU~Dortmund, Germany \\ \{firstname.lastname\}@tu-dortmund.de}
\date{}

\newcommand{\ie}{i.\,e., }

\newcommand{\Ex}[1]{\mathbb{E}\left[#1 \right]}
\newcommand{\Cov}[1]{Cov\left[#1 \right]}
\newcommand{\Prob}[1]{\mathbb{P}\left[#1 \right]}
\newcommand{\Var}[1]{Var\left[#1 \right]}
\newcommand{\polylog}{\text{polylog~}}
\newcommand{\poly}{\text{poly~}}
\begin{document}
\maketitle

\begin{abstract}
This paper presents an algorithm for estimating the weight of a maximum weighted matching by augmenting any estimation routine for the size of an unweighted matching. The algorithm is implementable in any streaming model including dynamic graph streams. We also give the first constant estimation for the maximum matching size in a dynamic graph stream for planar graphs (or any graph with bounded arboricity) using $\tilde{O}(n^{4/5})$ space which also extends to weighted matching. Using previous results by Kapralov, Khanna, and Sudan (2014) we obtain a $\mathrm{polylog}(n)$ approximation for general graphs using $\mathrm{polylog}(n)$ space in random order streams, respectively. In addition, we give a space lower bound of $\Omega(n^{1-\varepsilon})$ for any randomized algorithm estimating the size of a maximum matching up to a $1+O(\varepsilon)$ factor for adversarial streams.
\end{abstract}

\section{Introduction}
%\vspace{-0.2cm}
Large graph structures encountered in social networks or the web-graph have become focus of analysis both from theory and practice. To process such large input, conventional algorithms often require an infeasible amount of running time, space or both, giving rise to other models of computation. Much theoretical research focuses on the streaming model where the input arrives one by one with the goal of storing as much information as possible in small, preferably polylogarithmic, space.
%~\cite{Mut05} 
Streaming algorithms on graphs were first studied by Henzinger et al.~\cite{HRR98}, who showed that even simple problems often admit no solution with such small space requirements. The semi-streaming model \cite{FKM05} where the stream consists of the edges of a graph and the algorithm is allowed $O(n\cdot \text{polylog}(n))$ space and allows few (ideally just one) passes over the data relaxes these requirements and has received considerable attention. Problems studied in the semi-streaming model include sparsification,
%~\cite{AhG09,KeL13}
spanners,
%~\cite{Bas08, FKM08,Elk11}
connectivity, minimum spanning trees,
%~\cite{FKM08}
counting triangles
%~\cite{BKS02,BFL06,JoG05} 
and matching, for an overview we refer to a recent survey by McGregor~\cite{McG14}. Due to the fact that graphs motivating this research are dynamic structures that change over time there has recently been research on streaming algorithms supporting deletions. We now review the literature on streaming algorithms for matching and dynamic streams.

\paragraph{\textbf{Matching}}
Maintaining a $2$ approximation to the maximum matching (MM) in an insertion-only stream can be straightforwardly done by greedily maintaining a maximal matching \cite{FKM05}. Improving on this algorithm turns out to be difficult as Goel et al.~\cite{GKK12} showed that no algorithm using $\tilde{O}(n)$ space can achieve an approximation ratio better than $\frac{3}{2}$ which was improved by Kapralov to $\frac{e}{e-1}$~\cite{Kap13}. Konrad et al.~\cite{KMM12} gave an algorithm using $\tilde{O}(n)$ space with an approximation factor of $1.989$ if the edges are assumed to arrive in random order. 
%For distributed matching on $k$ sites, Huang et al.~\cite{HRVZ15} gave a lower bound of $\Omega(k \cdot n/\alpha^2)$ for any approximation factor $\alpha>1$. 
For weighted matching (MWM), a series of results have been published \cite{FKM05,McG05,ELM11,Zel12,ELS13} with the current best bound of $4+\varepsilon$ being due to Crouch and Stubbs~\cite{CrS14}.

To bypass the natural $\Omega(n)$ bound required by any algorithm maintaining an approximate matching, recent  research has begun to focus on estimating the size of the maximum matching. Kapralov et al.~\cite{KKS14} gave a polylogrithmic approximate estimate using polylogarithmic space for random order streams. 
For certain sparse graphs including planar graphs, Esfandiari et al.~\cite{EHL15} describe how to obtain a constant factor estimation using $\tilde{O}(n^{2/3})$ space in a single pass and $\tilde{O}(\sqrt{n})$ space using two passes or assuming randomly ordered streams. The authors also gave a lower bound of $\Omega(\sqrt{n})$ for any approximation better than $\frac{3}{2}$.

%\vspace{-0.4cm}
\paragraph{\textbf{Dynamic Streams}}
In the turnstile model, the stream consists of a sequence of additive updates to a vector.
%There are two dynamic streaming models: \emph{point insertion model} where individual points are added and deleted over time and the \emph{turnstile model} where updates to a vector or matrix arrive over time. There exist algorithms for  clustering~\cite{FS05}, extent approximations~\cite{AnN12} and various geometric problems~\cite{Ind04,FIS08} for the former model.  
Problems studied in this model include numerical linear algebra problems such as regression and low-rank approximation, and maintaining certain statistics of a vector like frequency moments, heavy hitters or entropy. 
%Indeed, the seminal paper by Alon, Matias and Szegedy~\cite{AMS99} describes a dynamic algorithm for approximating the second frequency moment. 
Linear sketches have proven to be the algorithmic technique of choice and might as well be the only algorithmic tool able to efficiently do so, see Li, Nguyen and Woodruff~\cite{NgW14}. Dynamic graphs as introduced and studied by Ahn, Guha and McGregor~\cite{AhG09,AGM12,AGM12a,AGM13a} are similar to, but weaker than turnstile updates. Though both streaming models assume update to the input matrix, there usually exists a consistency assumption for streams, i.e. at any given time the multiplicity of an edge is either $0$ or $1$ and edge weights cannot change arbitrarily but are first set to $0$ and then reinserted with the desired weight. The authors extend some of the aforementioned problems such as connectivity, sparsification and minimum spanning trees to this setting.
Recent results by Assadi et al.~\cite{AKLY15} showed that approximating matchings in dynamic streams is hard by providing a space lower bound of $\Omega(n^{2-3\varepsilon})$ for approximating the maximum matching within a factor of $\tilde{O}(n^{\varepsilon})$. Simultaneously, Konrad \cite{Kon15} showed a similar but slightly weaker lower bound of $\Omega(n^{3/2-4\varepsilon})$. Both works presented an algorithm with an almost matching upper bound on the space complexity of $\tilde{O}(n^{2-2\varepsilon})$ \cite{Kon15} and $\tilde{O}(n^{2-3\varepsilon})$ \cite{AKLY15}. Chitnis et al.~\cite{CCEHMMV15} gave a streaming algorithm using $\tilde{O}(k^2)$ space that returns an exact maximum matching under the assumption that the size is at most $k$. It is important to note that all these results actually compute a matching. In terms of estimating the size of the maximum matching, Chitnis et al.~\cite{CCEHMMV15} extended the estimation algorithms for sparse graphs from \cite{EHL15} to the settings of dynamic streams using $\tilde{O}(n^{4/5})$ space.
A bridge between dynamic graphs and the insertion-only streaming model is the sliding window model studied by Crouch et al.~\cite{CMS13}. The authors give a $(3+\varepsilon)$-approximation algorithm for maximum matching.

The $p$-Schatten norm of a matrix $A$ is defined as the $\ell_p$-norm of the vector of singular values.
%$||A||_{S_p}=\left(\sum_{i=1}^n \sigma_i^p\right)^{\frac{1}{p}}$ where $\sigma_i$ is the $i$-th singular value of $A$. 
%Common special cases include the Frobenius norm $||A||_{S_2}=\sqrt{\sum_{i,j} A_{ij}^2}$, 
%%the nuclear norm $||A||_{S_1}=\sum_{i=1}^n |\sigma_i|$ 
%and the rank $||A||_{S_0} = |\{\sigma_i|\sigma_i\neq 0\}|$. 
It is well known that computing the maximum matching size is equivalent to computing the rank of the Tutte matrix~\cite{Tut47,Lov79} (see also Section \ref{sec:appl}). Estimating the maximum matching size therefore is a special case of estimating the rank or $0$-Schatten norm of a matrix. Li, Nguyen and Woodruff gave strong lower bounds on the space requirement for estimating Schatten norms in dynamic streams \cite{LNW14b}. 
%showed that, with the exception of the Frobenius norm, approximating Schatten norms is far more difficult than approximating the vector frequency counterpart~\cite{LNW14b}. 
Any estimation of the rank within any constant factor is shown to require $\Omega(n^2)$ space when using bi-linear sketches and $\Omega(\sqrt{n})$ space for general linear sketches. 
%It should be noted that with the exception of $||A||_{S_2}$, all known algorithms for Schatten norms and graph problems in dynamic streams are based on bi-linear sketches.
%For instance Johnson-Lindenstrauss type embeddings with target dimension $O(k)$ perserve subspaces up to dimension $k$. Embedding both rows and colums into $O(k)$ dimensions and assuming at most $\log n$ bits per entry results in a $O(k^2\log n)$ space algorithm for deciding whether the rank is greater than $k$ or not, see Clarkson and Woodruff~\cite{CW09} and Andoni and Nguyen~\cite{AnN13}.
%\vspace{-0.3cm}
\begin{table}[h]
\begin{center}
\begin{tabular}{l r|c|c|c|c}
 & Reference & Graph class & Streaming model & Approx. factor & Space \\
 \hline
 \hline
\textbf{MM:} & Greedy  & General & Adversarial & $2$ & $O(n)$ \\ 
& \cite{KKS14} & General & Random & $\mathrm{polylog}(n)$ & $\mathrm{polylog}(n)$ \\
& \cite{EHL15} & Trees & Adversarial & $2+\varepsilon$ & $\tilde{O}(\sqrt{n})$ \\
& \cite{EHL15} & Bounded arboricity & Adversarial & $O(1)$ & $\tilde{O}(n^{2/3})$\\
& here & Trees & Dynamic & $2+\varepsilon$ & $O(\frac{\log^2 n}{\varepsilon^2})$\\
& here & Bounded arboricity & Dynamic & $O(1)$ & $\tilde{O}(n^{4/5})$\\
\hline
& \cite{EHL15} & Forests & Adversarial & $\frac{3}{2}-\varepsilon$ & $\Omega (\sqrt{n})$\\
& here & General & Adversarial & $1+O(\varepsilon)$ & $\Omega \left(n^{1-\varepsilon}\right)$\\
\hline
\hline
\textbf{MWM:} & \cite{CrS14} & General & Adversarial & $4+\varepsilon$ & $O(n\log^2 n)$ \\
& here & General & Random & $\mathrm{polylog}(n)$ & $\mathrm{polylog}(n)$ \\
& here & Bounded arboricity & Dynamic & $O(1)$ & $\tilde{O}(n^{4/5})$
\end{tabular}
\end{center}
\caption{Results for estimating the size (weight) of a maximum (weighted) matching in data streams.}\label{tbl:overview}
%\vspace{-1cm}
\end{table}
\subsubsection*{Techniques and Contribution}
%Motivated by an open problem from the Bertinoro 2014 list~\cite{sublinearopen}, we investigate estimation routines for insertion only streams~\cite{EHL15} and distributed matching~\cite{CzygrinowHS09} in the context of dynamic data streams. 
Table \ref{tbl:overview} gives an overview of our results in comparison to previously known algorithms and lower bounds. Our first main result (Section~\ref{sec:weighted}) is an approximate estimation algorithm for the maximum weight of a matching. We give a generic procedure using any unweighted estimation as black box. In particular:
\begin{theorem}[\rm informal version]
Given a $\lambda$-approximate estimation using $S$ space, there exists an $O(\lambda^4)$-approximate estimation algorithm for the weighted matching problem using $O(S \cdot \log n)$ space.
%For every $\lambda$-approximate estimation algorithm for the unweighted matching problem using $S$ space, there exists an $O(\lambda^4)$-approximate estimation algorithm for the weighted matching problem using $O(S \cdot \log n)$ space.
\end{theorem}

%Combining this theorem with a recent result of Kapralov et al.~\cite{KKS14} gives a $\mathrm{polylog}(n)$ space and $\mathrm{polylog}(n)$-approximate estimate for weighted matching in random order streams.

The previous algorithms for weighted matchings in insertion only streams analyzed in~\cite{FKM05,McG05,ELM11,Zel12} extend the greedy approach by a charging scheme. If edges are mutually exclusive, the new edge will be added if the weight of the matching increases by a given threshold, implicitly partitioning the edges into sets of geometrically increasing weights. We use a similar scheme, but with a twist: Single edge weights cannot be charged to an edge with larger weight as estimation routines do not necessarily give information on distinct edges. However, entire matchings can be charged as the contribution of a specific range of weights $r$ can only be large if these edges take up a significant part of any maximum matching in the subgraph containing only the edges of weight at least $r$. For analysis, we use a result on parallel algorithms by Uehara and Chen \cite{UeC00}. We show that the weight outputted by our algorithm is close to the weight of the matching computed by the authors, implying an approximation to the maximum weight. 

We can implement this algorithm in dynamic streams although at submission, we were unaware of any estimations for dynamic streams. Building on the work by Esfandiari et al. \cite{EHL15}, we give a constant estimation on the matching size in bounded arboricity graphs. The main obstacle to adapt their algorithms for bounded arboricity graphs is that they maintain a small size matching using the greedy algorithm which is hard for dynamic streams. Instead of maintaining a matching, we use the Tutte matrix to get a 1-pass streaming algorithm using $\tilde{O}(n^{4/5})$ space, which immediately extends to weighted matching. Similar bounds have been obtained independently by Chitnis et al. \cite{CCEHMMV15}.

%In addition, we also present an algorithm maintaining a small matching in two passes using sublinear space. Surprisingly, in terms of space usage both the 2-pass algorithm using the Tutte matrix and the 2-pass algorithm maintaining a small matching need $\tilde{O}(n^{2/3})$ space in the end.

Our lower bound (Section~\ref{sec:lower}) is proven via reduction from the Boolean Hidden Hypermatching problem introduced by Verbin and Yu~\cite{VerW11}. In this setting, two players Alice and Bob are given a binary $n$-bit string and a perfect $t$-hypermatching on $n$ nodes, respectively. Bob also gets a binary string $w$. The players are promised that the parity of bits corresponding to the nodes of the $i$-th hypermatching either are equal to $w_i$ for all $i$ or equal to $1-w_i$ for all $i$ and the task is to find out which case holds using only a single round of communication. We construct a graph consisting of a $t$-clique for each hyperedge of Bob's matching and a single edge for each bit of Alice's input that has one node in common with the $t$-cliques. Then we show that approximating the matching size within a factor better than $1+O(1/t)$ can also solve the Boolean Hidden Hypermatching instance. Using the lower bound of $\Omega(n^{1-1/t})$ from \cite{VerW11} we have
%\vspace{-0.5cm}
\begin{theorem}[\rm informal version]
\label{thm:lowerbound}
Any $1$-pass streaming algorithm approximating the size of the maximum matching matching up to an $(1+O(\varepsilon))$ factor requires $\Omega(n^{1-\varepsilon})$ bits of space.
\end{theorem}
%\vspace{-0.1cm}
This lower bound also implies an $\Omega(n^{1-\varepsilon})$ space bound for $1+O(\varepsilon)$ approximating the rank of a matrix in data streams which also improves the $\Omega(\sqrt{n})$ bound by Li, Nguyen, and Woodruff \cite{LNW14b} for linear sketches.
%\vspace{-0.3cm}
\subsection{Preliminaries}
\label{sec:prelim}
%\vspace{-0.2cm}
We use $\tilde{O}(f(n))$ to hide factors polylogarithmic in $f(n)$. Any randomized algorithm succeeding with high probability has at least $1-1/n$ chance of success.
Graphs are denoted by $G(V,E,w)$ where $V$ is the set of $n$ nodes, $E$ is the set of edges and $w:E\rightarrow\mathbb{R}^+$ is a weight function. 
%We omit $w$ for unweighted graphs. For a subset of nodes $S\subseteq V$, we denote by $E(S)$ the edges of the subgraph induced by $S$. 
Our estimated value $\widehat{M}$ is a $\lambda$-approximation to the size of the maximum matching $M$ if $\widehat{M} \leq |M|\leq \lambda\widehat{M}$.

\setcounter{theorem}{0}

%
%\vspace{-0.3cm}
%
%\begin{algorithm}[ht]
%\label{alg2}
%\caption{\bf{Weighted Matching Approximation($V,\bigcup_{i=0}^t E_i)$}}
%
%%\LinesNumbered
%\SetKwData{I}{i}
%\SetKwInOut{Input}{input}
%\Input{$G=(V,\bigcup_{i=0}^t E_i)$}
%$\widehat{R},\widehat{S},weight \leftarrow 0$, $last \leftarrow t$\;
%$\widehat{R_t},\widehat{S_t} \leftarrow \textbf{Unweighted Matching Approximation}(V, E_t)$ \\  %\mbox{}\phantom{\textbf{$\widehat{R_t},\widehat{S_t} \leftarrow$} \itshape(}$\textbf{Approximation}(V, E_t)$\;
%\For{$i=t-1 ${\bf{ to }$0$}}{
%$ \widehat{S_i} \leftarrow \textbf{Unweighted Matching Approximation}(V, \bigcup_{j=i}^t E_j)$ \\  %\mbox{}\phantom{\textbf{$MatchSize \leftarrow $} \itshape(}$\textbf{Approximation}(V, \bigcup\limits_{j=i}^t E_j)$\;
%\eIf(\tcp*[f]{add current index $i$ to $K$}){$\widehat{S_i} > \widehat{S_{last}} \cdot T$}{
%%$\widehat{R_i} \leftarrow \widehat{S_i} - \widehat{S_{last}}$ 
%\If(\tcp*[f]{add current index $i$ to $J$}){$\widehat{R_i} \geq c \cdot \widehat{R}_{last}$}{
%$\widehat{R_i} \leftarrow \widehat{S_i} - \widehat{S_{last}}$ \;
%$last\leftarrow i$
%}
%}{
%$\widehat{R_i} \leftarrow 0$\;
%}
%}
%\For{$i=t ${\bf{ to }$0$}}{
%$weight \leftarrow weight + r_i \cdot \widehat{R_i}$\;
%}
%
%\Return $weight$\;
%\end{algorithm}
%
%\section{Weighted Matching}
%\label{sec:weighted}
%\vspace{-0.2cm}

\section{Weighted Matching}
\label{sec:weighted}
%Consider a graph $G = (V,E,w)$ with arbitrary edge weights $w(e) \in \mathbb{R}^+$. 
We start by describing the parallel algorithm by Uehara and Chen \cite{UeC00}, see Algorithm~\ref{alg:weighted_matching_ueharachen}. Let $\gamma > 1$ and $k > 0$ be constant. We partition the edge set by $t$ ranks where all edges $e$ in rank $i \in \lbrace 1, \ldots, t \rbrace$ have a weight $w(e) \in \left( \gamma^{i-1} \cdot \frac{w_{max}}{kN}, \gamma^{i} \cdot \frac{w_{max}}{kN} \right]$ where $w_{max}$ is the maximal weight in $G$. Let $G' = (V,E,w)$ be equal to $G$ but each edge $e$ in rank $i$ has weight $r_i := \gamma^{i}$ for all $i = 1, \ldots, t$. Starting with $i = t$, we compute an unweighted maximal matching $M_i$ considering only edges in  rank $i$ (in $G'$) and remove all edges incident to a matched node. Continue with $i-1$. The weight of the matching $M = \bigcup M_i$ is $w(M) = \sum_{i=1}^t r_i \cdot \vert M_i \vert $ and satisfies $w_G(M^*) \geq w_{G'}(M) \geq \frac{1}{2 \gamma} \cdot w_G(M^*)$ where $M^*$ is an optimal weighted matching in $G$. 
%The greedy algorithm for MWM simply adds an edge to the matching (and removes adjacent edges) if the weight of the matching increases. It is easy to see that this algorithm does not give any guarantee on the solution.
The previous algorithms~\cite{FKM05,McG05,ELM11,Zel12,CrS14} for insertion-only streams use a similar partitioning of edge weights. Since these algorithms are limited to storing one maximal matching (in case of \cite{CrS14} one maximal matching per rank), they cannot compute residual maximal matchings in each rank. However, by charging the smaller edge weights into the higher ones, the resulting approximation factor can be made reasonably close to that of Uehara and Chen. 
%Maintaining maximal matchings in insertion-only streams is straightforward, but since we aim to supply an estimator for any streaming model we are not able to compute a maximal matching (or an approximation thereof). 
Since these algorithms maintain matchings, they cannot have sublinear space in an insertion-only stream and they need at least $\Omega(n^{2-3\varepsilon})$ in a dynamic stream even when the maintained matching is only a $O(n^{\varepsilon})$ approximation (\cite{AKLY15}). Though the complexity for unweighted estimating unweighted matchings is not settled for any streaming model, there exist graph classes for which one can improve on these algorithms wrt space requirement. Therefore, we assume the existence of a black box $\lambda$-approximate matching estimation algorithm.

\begin{algorithm}[t]
\caption{Approximation of Weighted Matching from \cite{UeC00}}
\label{alg:weighted_matching_ueharachen}
\algorithmicrequire{ Graph $G=(V,E=\bigcup_{i=1}^t E_i)$}\\
\algorithmicensure{ Matching}
\begin{algorithmic}
\For{$i=t $ to $1$}
\State{Find a maximal matching $M_i$ in $G_i = (V,E_i)$.}
\State{Remove all edges $e$ from $E$ such that $e \in M_i$ or $e$ shares a node with an edge in $M_i$.}
\EndFor
\State{\Return $\bigcup_{i=1}^t M_i$}
\end{algorithmic}
\end{algorithm}

\subsubsection*{Algorithm and Analysis}

In order to adapt this idea to our setting, we need to work out the key properties of the partitioning and how we can implement it in a stream. The first problem is that we cannot know $w_{max}$ in a stream a priori and in a dynamic stream even maintaining $w_{max}$ is difficult. However, the appropriate partition of an inserted edge depends on $w_{max}$. Recalling the partitioning of Uehara and Chen, we disregard all edges with weight smaller than $\frac{w_{max}}{kN}$ which is possible because the contribution of these edges is at most $\frac{N}{2} \cdot \frac{w_{max}}{kN} = \frac{w_{max}}{2k} \leq \frac{OPT}{2k}$ where $OPT$ is the weight of an optimal weighted matching. Thus, we can only consider edges with larger weight and it is also possible to partition the set of edges in a logarithmic number of sets. Here, we use the properties that edge weights within a single partition set are similar and that $\frac{1}{\gamma} \leq \frac{w(e)}{w(e')} \leq \gamma$ for two edges $e \in E_i$ and $e' \in E_{i-1}$ with $i \in \lbrace 2, \ldots, t \rbrace$. These properties are sufficient to get a good approximation on the optimal weighted matching which we show in the next lemma. The proof is essentially the same as in \cite{UeC00}.

\begin{lemma}
\label{lem:properties_weight_partition}
Let $G = (V,E,w)$ be a weighted graph and $\varepsilon > 0$ be an approximation parameter. If a partitioning $E_1, \ldots, E_t$ of $E$ and a weight function $w': E \rightarrow \mathbb{R}$ satisfy
$$ \frac{1}{1+\varepsilon} \leq \frac{w'(e)}{w(e)} \leq 1 \text{ for all } e \in E \quad \text{ and } \quad \frac{w(e_1)}{w(e_2)} \leq 1+\varepsilon \quad \text{ and } \quad w(e) < w(e') $$
for all choices of edges $e_1,e_2 \in E_i$ and $e \in E_i, e' \in E_j$ with $i < j$ and $i,j \in \lbrace 1, \ldots, t \rbrace$ then Algorithm \ref{alg:weighted_matching_ueharachen} returns a matching $M = \bigcup_{i=1}^t M_i$ with 
$$ \frac{1}{2 (1+\varepsilon)^2} \cdot w(M^*) \leq w'(M) \leq w(M^*)$$ where $M^*$ is an optimal weighted matching in $G$. 
\end{lemma}
\begin{proof}
The first property $ \frac{1}{1+\varepsilon} \leq \frac{w'(e)}{w(e)} \leq 1$ for all $e \in E$ implies that $\frac{w(S)}{1+\varepsilon} \leq w'(S) \leq w(S)$ for every set of edges $S \subseteq E$. Thus, it remains to show that $\frac{1}{2 (1+\varepsilon)} \cdot w(M^*) \leq w(M) \leq w(M^*)$. Since $M^*$ is an optimal weighted matching, it is clear that $w(M) \leq w(M^*)$. For the lower bound, we distribute the weight of the edges from the optimal solution to edges in $M$. Let $e \in M^*$ and $i \in \lbrace 1, \ldots, t \rbrace$ such that $e \in E_i$. We consider the following cases:
\begin{enumerate}
\item $e \in M_i$: We charge the weight $w(e)$ to the edge itself.
\item $e \not\in M_i$ but at least one node incident to $e$ is matched by an edge in $M_i$: Let $e' \in M_i$ be an edge sharing a node with $e$. Distribute the weight $w(e)$ to $e'$. 
\item $e \not\in M_i$ and there is no edge in $M_i$ sharing a node with $e$: By Algorithm \ref{alg:weighted_matching_ueharachen}, there has to be an edge $e' \in M_j$ with $j > i$ which shares a node with $e$. We distribute the weight $w(e)$ to $e'$.
\end{enumerate}
Since $M^*$ is a matching, there can only be at most two edges from $M^*$ distributing their weights to an edge in $M$. We know that $\frac{w(e)}{w(e')}  \leq 1+\varepsilon$ for all choices of two edges $e,e' \in E_i$ with $i \in \lbrace 1, \ldots, t \rbrace$ which means that in the case 2.~we have $w(e) \leq (1+\varepsilon) \cdot w(e')$. In case 3.~it holds $w(e) < w(e')$. Thus, the weight distributed to an edge $e'$ in $M$ is at most $2(1+\varepsilon)w(e')$. This implies that $w(M^*) = \sum_{e \in M^*} w(e) \leq \sum_{e' \in M} 2(1+\varepsilon) \cdot w(e') = 2(1+\varepsilon) \cdot w(M)$ which concludes the proof.
\end{proof}

Using Lemma \ref{lem:properties_weight_partition}, we can partition the edge set in a stream in an almost oblivious manner: Let $(e_0,w(e_0))$ be the first inserted edge. Then an edge $e$ belongs to $E_i$ iff $2^{i-1} \cdot w(e_0) < w(e) \leq 2^i \cdot w(e_0)$ for some $i \in \mathbb{Z}$. For the sake of simplicity, we assume that the edge weights are in $\left[ 1, W \right]$. Then the number of sets is $\O(\log W)$. We would typically expect $W\in \poly n$ as otherwise storing weights becomes infeasible.

We now introduce a bit of notation we will use in the algorithm and throughout the proof. We partition the edge set $E = \bigcup_{i=0}^t E_i$ by $t+1 = O(\log W)$ ranks where the set $E_i$ contains all edges $e$ with weight $w(e) \in \left[ 2^{i}, 2^{i+1} \right)$. Wlog we assume $E_t \neq \emptyset$ (otherwise let $t$ be the largest rank with $E_t \neq \emptyset$). Let $G' = (V,E,w')$ be equal to $G$ but each edge $e \in E_i$ has weight $w'(e) = r_i := 2^{i}$ for all $i = 0, \ldots, t$. Let $M = \bigcup_{i=0}^t M_i$ be the matching computed by the partitioning algorithm and $S$ be a $(t+1)$-dimensional vector with $S_i=\sum_{j=i}^t |M_i|$.

Algorithm~\ref{alg:weighted_matching} now proceeds as follows: For every $i\in \{0,\ldots t\}$ the size of a maximum matching in $(V,\bigcup_{j=i}^t E_j)$ and $S_i$ differ by only a constant factor. Conceptually, we set our estimator $\widehat{S_i}$ of $S_i$ to be the approximation of the size of the maximum matching of $(V,\bigcup_{j=i}^t E_i)$ and the estimator of the contribution of the edges in $E_i$ to the weight of an optimal weighted matching is $\widehat{R_i} = \widehat{S_i}-\widehat{S_{i+1}}$. The estimator $\widehat{R_i}$ is crude and generally not a good approximation to $|M_i|$. What helps us is that if the edges $M_i$ have a significant contribution to $w(M)$, then $|M_i|\gg \sum_{j=i+1}^t |M_j|=S_{i+1}$. 
In order to detect whether the matching $M_i$ has a significant contribution to the objective value, we introduce two parameters $T$ and $c$. The first matching $M_t$ is always significant (and the simplest to approximate by setting $\widehat{R_t}=\widehat{S_t}$). For all subsequent matchings $i<t$, let $j$ be the most recent matching which we deemed to be significant. We require $\widehat{S_i}\geq T\cdot \widehat{S_{j}}$ and $\widehat{R_i}\geq c\cdot \widehat{R_j}$. If both criteria are satisfied, we use the estimator $\widehat{R_i}=\widehat{S_i}-\widehat{S_{j}}$ and set $i$ to be the now most recent, significant matching, otherwise we set $\widehat{R_i}= 0$. The final estimator of the weight is $\sum_{i=0}^t r_i\cdot \widehat{R_i}$. The next definition gives a more detailed description of the two sets of ranks which are important for the analysis.
\begin{definition}[Good and Significant Ranks]
\label{def:good_significant_ranks}
Let $\widehat{S}$ and $\widehat{R}$ be the vectors at the end of Algorithm \ref{alg:weighted_matching}. An index $i$ is called to be a \emph{good rank} if $\widehat{S_i} \neq 0$ and $i$ is a \emph{significant rank} if $\widehat{R_i} \neq 0$. We denote the set of good ranks by $I_{good}$ and the set of significant ranks by $I_{sign}$, \ie $I_{good}:=\left\{i\subseteq \{0,\ldots t\}~| \widehat{S_i} \neq 0 \right\}$ and $I_{sign}:=\left\{i\subseteq \{0,\ldots t\}~| \widehat{R_i} \neq 0 \right\}$.
We define $I_{good}$ and $I_{sign}$ to be in descending order and we will refer to the $\ell$-th element of $I_{good}$ and $I_{sign}$ by $I_{good}(\ell)$ and $I_{sign}(\ell)$, respectively. That means
$ I_{good}(1) > I_{good}(2) > \ldots > I_{good}(\vert I_{good} \vert)$ and $ I_{sign}(1) > I_{sign}(2) > \ldots > I_{sign}(\vert I_{sign} \vert)$.
We slightly abuse the notation and set $I_{sign}(|I_{sign}|+1)=0$. Let $D_1 := \vert M_t \vert$ and for $\ell \in \lbrace 2, \ldots, \vert I_{sign} \vert \rbrace$ we define the sum of the matching sizes between two significant ranks $I_{sign}(\ell)$ and $I_{sign}(\ell-1)$ where the smaller significant rank is included by 
$D_{\ell}:=\sum_{i = I_{sign}(\ell)}^{I_{sign}(\ell-1)-1}|M_i|$.
% \quad \text{ and } \quad D_{|I_{sign}|+1} := \sum_{i = 0}^{I_{sign}(|I_{sign}|)-1}|M_i| \quad \text{ if $0\notin I_{sign}$}. $$
\end{definition}

\begin{algorithm}[t]
\caption{Weighted Matching Approximation}
\label{alg:weighted_matching}
\algorithmicrequire{ Graph $G=(V,\bigcup_{i=0}^t E_i)$} with weights $r_i$ for edges in $E_i$\\
\algorithmicensure{ Estimator of the weighted matching}
\begin{algorithmic}
%\Input{$G=(V,\bigcup_{i=0}^t E_i)$}
\For{$i=t ${\bf{ to }$0$}}
\State{$\widehat{S_i} = \widehat{R_i} = 0$}
\EndFor
\State{$weight = 0$, $last = t$}
\State{$\widehat{R_t} =\widehat{S_t} = \textbf{Unweighted Matching Estimation}(V, E_t)$} %\mbox{}\phantom{\textbf{$\widehat{R_t},\widehat{S_t} \leftarrow$} \itshape(}$\textbf{Approximation}(V, E_t)$\;
\For{$i=t-1 ${\bf{ to }$0$}}
\State{$ \widehat{S_i} = \textbf{Unweighted Matching Estimation}(V, \bigcup_{j=i}^t E_j)$}  %\mbox{}\phantom{\textbf{$MatchSize \leftarrow $} \itshape(}$\textbf{Approximation}(V, \bigcup\limits_{j=i}^t E_j)$\;
\If{$\widehat{S_i} > \widehat{S_{last}} \cdot T$} \hfill $\vartriangleright$ Add current index $i$ to $I_{good}$
%\eIf(\tcp*[f]{add current index $i$ to $K$}){$\widehat{S_i} > \widehat{S_{last}} \cdot T$}{
%$\widehat{R_i} \leftarrow \widehat{S_i} - \widehat{S_{last}}$ 
\If{$\widehat{S_i} - \widehat{S_{last}} \geq c \cdot \widehat{R}_{last}$} \hfill $\vartriangleright$ Add current index $i$ to $I_{sign}$
\State{$\widehat{R_i} = \widehat{S_i} - \widehat{S_{last}}$}
\State{$last = i$}
\EndIf
\Else
\State{$\widehat{S_i} = 0$}
\EndIf
\EndFor
\State{\Return $\frac{2}{5}\sum\limits_{i=0}^t r_i \cdot \widehat{R_i}$}
%\For{$i=t ${\bf{ to }$0$}}
%\State{$weight = weight + r_i \cdot \widehat{R_i}$}
%\EndFor
%\State{\Return $weight$}
\end{algorithmic}
\end{algorithm}

In the following, we subscript indices by $s$ for significant ranks and by $g$ for good ranks for the sake of readability.
Looking at Algorithm \ref{alg:weighted_matching} we can proof some simple properties of $I_{good}$ and $I_{sign}$.
\begin{lemma}
\label{lem:good_significant_ranks}
Let $I_{good}$ and $I_{sign}$ be defined as in Definition \ref{def:good_significant_ranks}. Then 
\begin{enumerate}
\item $I_{good}(1) = I_{sign}(1) = t$ and $I_{sign} \subseteq I_{good}$.
\item For every good rank $i_g \in I_{good}$ there is an $\ell \in \lbrace 0, \ldots, \vert I_{sign} \vert \rbrace$ such that $I_{sign}(\ell) > i_g \geq I_{sign}(\ell+1)$ and $\widehat{S_{i_g}}> T\cdot \widehat{S_{I_{sign}(\ell)}}$.
\item For every $i_s, i_s' \in I_{sign}$ with $i_s < i_s'$ it holds $\widehat{S_{i_s}} > T \cdot \widehat{S_{i_s'}}$.
\item For any $i_s \in I_{sign}$ and $i_s' \in I_{sign}$ with $i_s' < i_s$ it is $\widehat{R_{i_s'}} > c \cdot \widehat{R_{i_s}}$.
\item For any $i_s \in I_{sign}$ and $i_g \in I_{good}$ with $i_g < i_s$ it is $\widehat{S_{i_g}} > T \cdot \widehat{S_{i_s}}$.
%$I_{good} = I_{sign} \bigcup \{i\in \{0,\ldots t\}~|~ \exists \ell\in I_{sign}: ~\widehat{S_i}>T\cdot S_{I_{sign}(\ell)} \wedge I_{sign}(\ell+1) > i > I_{sign}(\ell) \}$
\end{enumerate}
\end{lemma}
\begin{proof}
~

\begin{enumerate}
\item It is clear that $I_{sign} \subseteq I_{good}$. Since we assumed that $E_t \neq \emptyset$, there is a nonempty matching in $E_t$ which means that $\widehat{S_t} = \widehat{R_t} > 0$. 
\item Let $\ell$ be the position of $last$ in $I_{sign}$ where $last$ is the value of the variable in Algorithm \ref{alg:weighted_matching} during the iteration $i = i_g$. Then $I_{sign}(\ell) > i_g \geq I_{sign}(\ell+1)$ (recall that we defined $I_{sign}(\vert I_{sign} \vert + 1) = 0$). Since $i_g$ is good, it is $\widehat{S}_{i_g} > T \cdot \widehat{S_{last}} = \widehat{S_{I_{sign(\ell)}}}$.
\item Since significant ranks are also good, we can apply 2.~to get $\widehat{S_{I_{sign}(\ell+1)}}> T\cdot \widehat{S_{I_{sign}(\ell)}}$ where $I_{sign}(\ell+1) < I_{sign}(\ell)$. By transitivity this implies the statement.
\item  For every $i_s \in I_{sign}$ we have $\widehat{R_{i_s}} \geq c \cdot \widehat{R_{last}}$ where $last$ is the value of the variable in Algorithm $\ref{alg:weighted_matching}$ in iteration $i = i_s$. By definition it is $last \in I_{sign}$ and $last > i_s$. Therefore, it holds $\widehat{R_{I_{sign}(\ell+1)}} > c \cdot \widehat{R_{I_{sign}(\ell)}}$ for every $\ell \in \lbrace 0, \ldots, \vert I_{sign} \vert-1 \rbrace$ which implies the statement.
\item Using 2.~we know that $\widehat{S_{i_g}}> T\cdot \widehat{S_{I_{sign}(\ell)}}$ for some $\ell \in \lbrace 0, \ldots, \vert I_{sign} \vert \rbrace$. If $i_s$ is equal to $I_{sign}(\ell)$ then we are done. Otherwise, we have $i_s > I_{sign}(\ell)$ and we can use 3.~to get $\widehat{S_{i_g}} > T \cdot \widehat{S_{I_{sign}(\ell)}} > T \cdot \widehat{S_{i_s}}$.
\end{enumerate}
\end{proof}
Now, we have the necessary notations and properties of good and significant ranks to proof our main theorem.
\begin{theorem}
\label{thm:weighted}
Let $G = (V,E,w)$ be a weighted graph where the weights are from $[1,W]$. Let $A$ be an algorithm that returns an $\lambda$-estimator $\widehat{M}$ for the size of a maximum matching $M$ of a graph with $1/\lambda \cdot \vert M \vert \leq \widehat{M} \leq \vert M \vert$ with failure probability at most $\delta$ and needs space $S$. If we partition the edge set into sets $E_0,\ldots,E_t$ with $t = \lfloor \log W \rfloor$ where $E_i$ consists of all edges with weight in $[2^i,2^{i+1})$, set $r_i = 2^i$, and use $A$ as the unweighted matching estimator in Algorithm \ref{alg:weighted_matching}, then there are parameters $T$ and $c$ depending on $\lambda$ such that the algorithm returns an $O(\lambda^4)$-estimator $\widehat{W}$ for the weight of the maximum weighted matching with failure probability at most $\delta\cdot (t+1)$ using $O(S \cdot t)$ space, i.e. there is a constant $c$ such that
$ \frac{1}{c \lambda^4} \cdot w(M^*) \leq \widehat{W} \leq w(M^*) $
where $M^*$ is an optimal weighted matching.
\end{theorem}
\begin{proof}
In the following we condition on the event that all calls to the unweighted estimation routine succeed, which happens with probability at least $1-\delta\cdot (t+1)$.
The estimator returned by Algorithm \ref{alg:weighted_matching} can be written as $\sum_{\ell=1}^{|I_{sign}|} r_{I_{sign}(\ell)} \cdot \widehat{R_{I_{sign}(\ell)}}$. 
%Let $M = \bigcup_{i=0}^t M_i$ be the result of the partitioning algorithm of Uehara and Chen on the input $(V,\bigcup E_i)$. 
Using similar arguments as found in Lemma 4 of \cite{UeC00}, we have $\frac{1}{8} \cdot w(M^*) \leq \sum\limits_{i=0}^t r_i \vert M_i \vert \leq w(M^*)$.  
Thus, it is sufficient to show that $\sum_{\ell=1}^{|I_{sign}|} r_{I_{sgin}(\ell)} \cdot \widehat{R_{I_{sign}(\ell)}}$ is a good estimator for $\sum\limits_{i=0}^t r_i \vert M_i \vert$. We first consider the problem of estimating $D_{\ell}$, and then how to charge the matching sizes.
%\begin{enumerate}
%\item \textbf{(Estimation)} $\widehat{R_{I_{sign}(\ell)}}$ is a good estimator for $D_\ell$.
%\item \textbf{(Charging)} We show that $\sum\limits_{i=0}^t r_i \vert M_i \vert$ can be estimated by $\sum\limits_{\ell=1}^{\vert I_{sign} \vert}  r_{I_{sign}(\ell)} D_\ell$.
%\end{enumerate}

\subsubsection*{(1) Estimation of $D_\ell$}
%We start with two auxiliary lemmas. Let $M = \bigcup_{i=0}^t M_i$ be the result of Algorithm \ref{alg:weighted_matching_ueharachen} on the input $(V,\bigcup E_i)$. 

%Recall that we defined $S_i = \sum_{j=i}^t |M_i|$.

%The parameters $T$ and $c$ depend on $\lambda$ and will be specified later. We note that $J$ is a subset of $K$ and $J(1)=K(1)=t$. Furthermore for any $j\in J$ and $k\in K$ with $k < j$ we have $\widehat{S_k} > \widehat{S_j}\cdot T$ and for all $j'\in J$ with $j'<j$ we have $\widehat{R_{j'}}\geq \widehat{R_{j}}\cdot c$.
Since $\bigcup_{j=i}^t M_j$ is a maximal matching in $\bigcup_{j=i}^t E_j$, $\widehat{S_i}$ is a good estimator for $S_i$: 

\begin{lemma}
\label{lem:aux1}
For all $i \in \lbrace 0, \ldots, t \rbrace$ we have
$\displaystyle \frac{1}{\lambda}\cdot S_i \leq \widehat{S_i} \leq 2\cdot S_i$.
\end{lemma}
\begin{proof}
Let $F_{j}$ be the set of unmatched nodes after the iteration $j$ of Algorithm \ref{alg:weighted_matching_ueharachen}. Let $M^*$ be a maximum matching in $(V,\bigcup_{j=i}^t E_{j})$. $M_{j}$ is a maximal matching of $(V,E_{j}(F_{j}))$ and therefore $\bigcup_{j=i}^t M_{j}$ is a maximal matching of $(V,\bigcup_{j=i}^t E_{j})$. This allows us to apply the bounds of the $\lambda$-approximate estimation algorithm:
$$ \displaystyle\frac{1}{\lambda}\cdot S_i=\frac{1}{\lambda}\cdot \sum_{j = i}^t |M_{j}| \leq  \displaystyle\frac{1}{\lambda}\cdot \vert M^* \vert \leq \vert \widehat{S_i} \vert \leq \vert M^* \vert \leq 2 \cdot \sum_{j = i}^t |M_{j}| = 2\cdot S_i. $$ 
\end{proof}
Next, we show that for an index $i_g \in I_{good}$ the difference $\widehat{S_{i_g}}- \widehat{S_{I_{sign}(\ell)}}$ to the last significant rank is a good estimator for $\sum_{i=i_g}^{I_{sign}(\ell)-1} |M_i|$.

\begin{lemma}
\label{lem:K-Bounds1}
For all $i_g\in I_{good}$ with $I_{sign}(\ell+1)\leq i_g < I_{sign}(\ell)$ for some $\ell \in \lbrace 1, \ldots, \vert I_{sign} \vert \rbrace$ and $T = 8\lambda^2-2\lambda$,
\[\frac{1}{2\lambda}\cdot \sum_{i=i_g}^{I_{sign}(\ell)-1} |M_i| < \widehat{S_{i_g}} - \widehat{S_{I_{sign}(\ell)}} < \frac{5}{2}\cdot \sum_{i=i_g}^{I_{sign}(\ell)-1} |M_i|\]
and $\frac{1}{\lambda} \vert M_t \vert \leq \widehat{S_t} \leq 2 \vert M_t \vert$.
\end{lemma}
\begin{proof}
For all $i_g \in I_{good}$ with $I_{sign}(\ell+1)\leq i_g < I_{sign}(\ell)$ we have
\begin{eqnarray}
\nonumber
\sum_{i=i_g}^{I_{sign}(\ell)-1} |M_i| &=& S_{i_g} - S_{I_{sign}(\ell)} 
\underset{\text{Lem.~\ref{lem:aux1}}}{\geq}  \frac{1}{2}\cdot \widehat{S_{i_g}} - \lambda \cdot \widehat{S_{I_{sign}(\ell)}} \\
\nonumber
&\underset{\text{Lem.~\ref{lem:good_significant_ranks} (2)}}{>}&  \frac{T}{2}\cdot \widehat{S_{I_{sign}(\ell)}} - \lambda \cdot \widehat{S_{I_{sign}(\ell)}}
\underset{\text{Lem.~\ref{lem:aux1}}}{\geq}  \left( \frac{T}{2}-\lambda\right)\cdot \frac{1}{\lambda}\cdot S_{I_{sign}(\ell)} \\
\label{eq:blockbound}
&=& \frac{T-2\lambda}{2\lambda}\cdot S_{I_{sign}(\ell)},
\end{eqnarray}
Setting $T= 8\lambda^2-2\lambda$, we then obtain the following upper and lower bounds
\begin{eqnarray}
\nonumber
\widehat{S_{i_g}} - \widehat{S_{I_{sign}(\ell)}} &\underset{\text{Lem. \ref{lem:aux1}}}{\geq}&\frac{1}{\lambda}\cdot S_{i_g} - 2\cdot S_{I_{sign}(\ell)} = \frac{1}{\lambda}\sum_{i=i_g}^{I_{sign}(\ell)-1}|M_i|-\left(2-\frac{1}{\lambda}\right) \cdot S_{I_{sign}(\ell)}\\
\nonumber
&\underset{\text{Eq.}~\ref{eq:blockbound}}{>}&  \frac{1}{\lambda}\sum_{i=i_g}^{I_{sign}(\ell)-1}|M_i|-\left(2-\frac{1}{\lambda}\right) \cdot \frac{2\lambda}{T-2\lambda} \cdot \sum_{i=i_g}^{I_{sign}(\ell)-1} |M_i| 
\end{eqnarray}
\begin{eqnarray}
\nonumber
\Rightarrow\widehat{S_{i_g}} - \widehat{S_{I_{sign}(\ell)}}& =& \frac{1}{\lambda}\sum_{i=i_g}^{I_{sign}(\ell)-1}|M_i|-\left(\frac{2\lambda-1}{\lambda}\right) \cdot \frac{2\lambda}{8\lambda^2-4\lambda} \cdot \sum_{i=i_g}^{I_{sign}(\ell)-1} |M_i| \\
\nonumber
&=& \left(\frac{1}{\lambda}- \frac{2}{4\lambda}\right)\sum_{i=i_g}^{I_{sign}(\ell)-1} |M_i|\\
\nonumber
&=& \frac{1}{2\lambda}\cdot\sum_{i=i_g}^{I_{sign}(\ell)-1} |M_i|
\end{eqnarray}
and
\begin{eqnarray}
\nonumber
 \widehat{S_{i_g}} - \widehat{S_{I_{sign}(\ell)}} &\underset{\text{Lem. \ref{lem:aux1}}}{\leq}&2\cdot S_{i_g} - \frac{1}{\lambda}\cdot S_{I_{sign}(\ell)} = 2\cdot \sum_{i=i_g}^{I_{sign}(\ell)-1}|M_i|+\left(2-\frac{1}{\lambda}\right) \cdot S_{I_{sign}(\ell)}\\
\nonumber
&\underset{\text{Eq.}~\ref{eq:blockbound}}{<}&  2\sum_{i=i_g}^{I_{sign}(\ell)-1}|M_i| + \left(2-\frac{1}{\lambda}\right) \cdot \frac{2\lambda}{8\lambda^2-4\lambda} \cdot \sum_{i=i_g}^{I_{sign}(\ell)-1} |M_i| \\
\nonumber
&=& \left(2 + \frac{2}{4\lambda}\right)\sum_{i=i_g}^{I_{sign}(\ell-1)-1} |M_i|\\
\nonumber
&\leq& \frac{5}{2}\cdot \sum_{i=i_g}^{I_{sign}(\ell-1)-1} |M_i|,
\end{eqnarray}
where we used $\lambda\geq 1$ in the last inequality. Since $\vert M_t \vert = S_t$ and $\widehat{R_t} = \widehat{S_t}$, the last statement follows directly from Lemma \ref{lem:aux1}.
\end{proof}
From Lemma \ref{lem:good_significant_ranks}.1 we know that $I_{sign} \subseteq I_{good}$ which together with the last Lemma \ref{lem:K-Bounds1} implies that $\widehat{R_{I_{sign}(\ell)}}$ is a good estimator for $D_\ell$.

\begin{corollary}
\label{cor:rjestimator_dlincreasing}
For $\ell \in \lbrace 1, \ldots, \vert I_{sign} \vert \rbrace$,
$\frac{1}{2\lambda}\cdot D_{\ell} \leq \widehat{R_{J(\ell)}} \leq \frac{5}{2}\cdot D_{\ell} $.
Furthermore, if $c > 5 \lambda$ then the values of the $D_\ell$ are exponentially increasing:
$$ D_{1}\leq \frac{5\lambda}{c} D_{2} \leq \ldots \leq \left(\frac{5\lambda}{c}\right)^{\vert I_{sign} \vert - 1} D_{\vert I_{sign} \vert - 1}. $$
\end{corollary}
\begin{proof}
Recall that for $\ell \in \lbrace 2,\ldots,\vert I_{sign} \vert \rbrace$ we defined $D_\ell = \sum_{i = I_{sign}(\ell)}^{I_{sign}(\ell-1)+1} \vert M_i \vert$. For $\ell = 1$ the value of $\widehat{R_{I_{sign}(1)}} = \widehat{S_t}$ is a good estimator for the size of the matching $M_t$ (which is equal to $D_1$) due to Lemma \ref{lem:aux1}. Since for $\ell \in \lbrace 2, \ldots, \vert I_{sign} \vert \rbrace$ it is $\widehat{R_{I_{sign}(\ell)}} = S_{I_{sign}(\ell)}-S_{I_{sign}(\ell-1)}$ and $I_{sign} \subseteq I_{good}$, the first statement is a direct implication of Lemma \ref{lem:K-Bounds1} by setting $i_g = I_{sign}(\ell)$.

For three adjoining significant ranks $I_{sign}(\ell+1),I_{sign}(\ell),I_{sign}(\ell-1)$ with $\ell \in \lbrace 2,\ldots, \vert I_{sign} \vert - 1 \rbrace$, we have
\begin{eqnarray}
\nonumber
\frac{1}{2\lambda}\cdot D_{\ell} = \frac{1}{2\lambda} \sum_{i=I_{sign}(\ell)}^{I_{sign}(\ell-1)-1} |M_i| & \underset{\text{Lem.~\ref{lem:K-Bounds1}}}{\leq} & \widehat{S_{I_{sign}(\ell)}} - \widehat{S_{I_{sign}(\ell-1)}} = \widehat{R_{I_{sign}(\ell)}} \\
\nonumber & \underset{\text{Lem.~\ref{lem:good_significant_ranks} (4)}}{\leq} & \frac{1}{c} \cdot \widehat{R_{I_{sign}(\ell+1)}} = \frac{1}{c} \cdot \left(\widehat{S_{I_{sign}(\ell+1)}} - \widehat{S_{I_{sign}(\ell)}}\right)\\
\nonumber
& \underset{\text{Lem.~\ref{lem:K-Bounds1}}}{\leq } & \frac{5}{2 c} \sum_{i=I_{sign}(\ell+1)}^{I_{sign}(\ell)-1} |M_i| = \frac{5}{2 c}\cdot D_{\ell+1}. 
\end{eqnarray}
Since $D_1 = \vert M_t \vert$ and $\widehat{R_{I_{sign}(1)}} = \widehat{R_{t}} = \widehat{S_t}$, Lemma \ref{lem:K-Bounds1} also implies that 
$$\frac{1}{2\lambda} \cdot D_1 \leq \frac{1}{\lambda} \cdot D_1 \leq \widehat{R_t} \leq \frac{1}{c} \cdot \widehat{R_{I_{sign}(2)}} \leq \frac{5}{2 c} \sum_{i=I_{sign}(2)}^{I_{sign}(1)-1} |M_i| = \frac{5}{2 c}\cdot D_{2}. $$
Thus, for $c>5\lambda$ the values of the $D_{\ell}$ are exponentially increasing:
$$
D_{1}\leq \frac{5\lambda}{c} D_{2} \leq \ldots \leq \left(\frac{5\lambda}{c}\right)^{\vert I_{sign} \vert - 1} D_{\vert I_{sign} \vert - 1}.
$$
\end{proof}

\subsubsection*{(2) The Charging Argument}
We show that the sum of the matching sizes between two significant ranks $I_{sign}(\ell+1)$ and $I_{sign}(\ell)$ is bounded by $O(\lambda \cdot T \cdot D_\ell) = O\left(\lambda \cdot T \cdot \sum_{i=I_{sign}(\ell)}^{I_{sign}(\ell-1)+1} \vert M_i \vert \right)$.

\begin{lemma}
\label{lem:charge_matchings}
Setting $c = \frac{2}{5} \cdot T + 5 \lambda$ in Algorithm \ref{alg:weighted_matching}. Then for $\ell \in \lbrace 1,\ldots,\vert I_{sign} \vert-1 \rbrace$,
$\displaystyle \sum_{i=I_{sign}(\ell+1)+1}^{I_{sign}(\ell)-1}|M_i| \leq (2\lambda\cdot T + 25\lambda^2)\cdot D_{\ell}$
and $\displaystyle \sum_{i=0}^{I_{sign}(|I_{sign}|)-1}|M_i| \leq (2\lambda\cdot T + 25\lambda^2)\cdot D_{|I_{sign}|}$ if $0 \not\in I_{sign}$.
\end{lemma}
\begin{proof}
For the proof of the first inequality, let $i_g \in I_{good}$ be minimal such that $I_{sign}(\ell+1)<i_g<I_{sign}(\ell)$ for $\ell \in \lbrace 1,\ldots,\vert I_{sign} \vert-1 \rbrace$. If such a good rank does not exist, set $i_g = -1$. We distinguish between two cases. Note that $c = \frac{5}{2} \cdot T + 5 \lambda > 5\lambda$.
\begin{description}
\item[Case 1: $i_g = I_{sign}(\ell+1)+1$.] For the sake of simplicity, we abuse the notation and set $\widehat{S_{I_{sign}(0)}} = 0$ such that $\widehat{R_{I_{sign}(\ell)}} = \widehat{S_{I_{sign}(\ell)}}-\widehat{S_{I_{sign}(\ell-1)}}$ also holds for $\ell = 1$. Using Lemma~\ref{lem:K-Bounds1} we have
\begin{eqnarray}
\nonumber
\sum_{i=I_{sign}(\ell+1)+1}^{I_{sign}(\ell)-1}|M_i| & = & \sum_{i=i_g}^{I_{sign}(\ell)-1}|M_i| \underset{\text{Lem.~\ref{lem:K-Bounds1}}}{\leq} 2\lambda\cdot\left( \widehat{S_{i_g}} - \widehat{S_{I_{sign}(\ell)}}\right)\\
\nonumber & \underset{i_g \not\in I_{sign}}{\leq} & 2\lambda c\cdot \widehat{R_{I_{sign}(\ell)}} = 2\lambda\cdot c \cdot \left( \widehat{S_{I_{sign}(\ell)}}-\widehat{S_{I_{sign}(\ell-1)}} \right) \\
& \underset{\text{Lem.~\ref{lem:K-Bounds1}}}{\leq} & 5\lambda\cdot c\cdot \sum_{i=I_{sign}(\ell)}^{I_{sign}(\ell-1)-1} |M_i| = 5\cdot c\cdot D_{\ell} \label{eq:case1a}
\end{eqnarray}

\item[Case 2: $i_g\neq I_{sign}(\ell+1)+1$.] In this case $\widehat{S_{I_{sign}(\ell+1)+1}}\leq T\cdot \widehat{S_{I_{sign}(\ell)}}$. Thus
\begin{eqnarray}
\nonumber
\sum_{i=I_{sign}(\ell+1)+1}^{I_{sign}(\ell)-1}|M_i|&\leq &  S_{I_{sign}(\ell+1)+1} \underset{\text{Lem.~\ref{lem:aux1}}}{\leq} \lambda\cdot \widehat{S_{I_{sign}(\ell+1)+1}} \\
\nonumber
&\leq & \lambda\cdot T\cdot \widehat{S_{I_{sign}(\ell)}} \underset{\text{Lem.~\ref{lem:aux1}}}{\leq} 2\lambda \cdot T\cdot S_{I_{sign}(\ell)} = 2\lambda\cdot T\cdot \sum_{i=1}^{\ell} D_{i}\\
&\underset{\text{Cor.~\ref{cor:rjestimator_dlincreasing}}}{\leq} & 2\lambda \cdot T\cdot D_{\ell}\cdot \sum_{i=1}^{\ell} \left(\frac{5\lambda}{c}\right)^{i} \leq 2\lambda \cdot T\cdot D_{\ell}\cdot\frac{1}{1-\frac{5\lambda}{c}}
\label{eq:case2a}
%\Leftrightarrow \sum_{i=I_{sign}(\ell+1)+1}^{I_{sign}(\ell)-1}|M_i| &\leq & \frac{2\lambda\cdot T}{1-\frac{5\lambda}{c}} \cdot D_{\ell}.
\end{eqnarray}
\end{description}
Combining the inequalities~\ref{eq:case1a} and~\ref{eq:case2a}, we have $\sum_{i=I_{sign}(\ell+1)+1}^{I_{sign}(\ell)-1}|M_i| \leq \linebreak \max\left\{5\lambda \cdot c,\frac{2\lambda\cdot T}{1-\frac{5\lambda}{c}}\right\}\cdot D_{\ell}$ which simplifies to
\begin{equation*}
\label{eq:Sbound1}
\sum_{i=I_{sign}(\ell+1)+1}^{I_{sign}(\ell)-1}|M_i| \leq (2\lambda\cdot T + 25\lambda^2)\cdot D_{\ell} \qquad \text{ for } \ell \in \lbrace 1,\ldots,\vert I_{sign} \vert-1 \rbrace.
\end{equation*}
If $0 \not\in I_{sign}$ we can do the same arguments to bound $\sum_{i=0}^{I_{sign}(|I_{sign}|)-1}|M_i|$ by $(2\lambda\cdot T + 25\lambda^2)\cdot D_{|I_{sign}|}$. 
Let $i_g \in I_{good}$ be minimal such that $0\leq i_g < I_{sign}(|I_{sign}|)$. Again, we distinguish between two cases.

\begin{description}
\item[Case 1: $i_g = 0$.] Using Lemma~\ref{lem:K-Bounds1} we have
\begin{eqnarray*}
\nonumber
\frac{1}{2\lambda}\cdot\sum_{i=0}^{I_{sign}(|I_{sign}|)-1}|M_i| & \underset{\text{Lem.~\ref{lem:K-Bounds1}}}{\leq} &  \widehat{S_{0}} - \widehat{S_{I_{sign}(|I_{sign}|)}}\\
\nonumber & \underset{0 \not\in I_{sign}}{\leq} & c\cdot \widehat{R_{I_{sign}(|I_{sign}|)}} = c \cdot \left( \widehat{S_{I_{sign}(|I_{sign}|)}}-\widehat{S_{I_{sign}(|I_{sign}|-1)}} \right) \\
\nonumber & \underset{\text{Lem.~\ref{lem:K-Bounds1}}}{\leq} & \frac{5}{2}\cdot c\cdot \sum_{i=I_{sign}(|I_{sign}|)}^{I_{sign}(|I_{sign}|-1)-1} |M_i| = \frac{5}{2}\cdot c\cdot D_{|I_{sign}|} \\
\Leftrightarrow \sum_{i=0}^{I_{sign}(|I_{sign}|)-1}|M_i| & \leq & 5\lambda \cdot c \cdot D_{|I_{sign}|}
\end{eqnarray*}
\item[Case 2: $i_g \neq 0$.] In this case $\widehat{S_{0}}\leq T\cdot \widehat{S_{I_{sign}(|I_{sign}|)}}$. Thus
\begin{eqnarray*}
\nonumber
\frac{1}{\lambda}\cdot \sum_{i=0}^{I_{sign}(|I_{sign}|)-1}|M_i|&\leq & \frac{1}{\lambda}\cdot S_{0} \underset{\text{Lem.~\ref{lem:aux1}}}{\leq} \widehat{S_{0}} \leq T\cdot \widehat{S_{I_{sign}(|I_{sign}|)}} \\
\nonumber
& \underset{\text{Lem.~\ref{lem:aux1}}}{\leq} & 2 \cdot T\cdot S_{I_{sign}(|I_{sign}|)} = 2\cdot T\cdot \sum_{i=1}^{|I_{sign}|} D_{i}\\
\nonumber
&\underset{\text{Cor.~\ref{cor:rjestimator_dlincreasing}}}{\leq} & 2 \cdot T\cdot D_{|I_{sign}|}\cdot \sum_{i=1}^{|I_{sign}|} \left(\frac{5\lambda}{c}\right)^{i} \leq 2 \cdot T\cdot D_{|I_{sign}|}\cdot\frac{1}{1-\frac{5\lambda}{c}}\\
\Leftrightarrow \sum_{i=0}^{I_{sign}(|I_{sign}|)-1}|M_i| &\leq & \frac{2\lambda\cdot T}{1-\frac{5\lambda}{c}} \cdot D_{|I_{sign}|}.
\end{eqnarray*}
\end{description}
Now, with the same $c = \frac{2}{5}\cdot T +5\lambda$ as before we have 
\begin{equation*}
\sum_{i=0}^{I_{sign}(|I_{sign}|)-1}|M_i| \leq (2\lambda\cdot T + 25\lambda^2)\cdot D_{|I_{sign}|}.
\end{equation*} 
%\qed
\end{proof}
We use Lemma~\ref{lem:charge_matchings} to show that $w(M)$ is bounded in terms of $\sum_{\ell=1}^{|I_{sign}|}r_{I_{sign}(\ell)}\cdot D_{\ell} $:
\begin{eqnarray}
\sum\limits_{i=0}^t r_i \cdot \vert M_i \vert  & \geq & \sum_{\ell=1}^{|I_{sign}|}r_{I_{sign}(\ell)}\cdot D_{\ell} \label{eq:OPTlower1}\\
\sum\limits_{i=0}^t r_i \cdot \vert M_i \vert 
& \leq & (1+2\lambda\cdot T +25\lambda^2)\cdot \sum_{\ell=1}^{|I_{sign}|}r_{I_{sign}(\ell)}\cdot D_{\ell}.\label{eq:OPTupper1}
\end{eqnarray} 

%What remains to show is that $\sum_{\ell=1}^{|J|}r_{J(\ell)}\cdot \widehat{R_{J(\ell)}}$ is a good estimation of \linebreak $\sum_{\ell=1}^{|J|+1}r_{J(\ell)}\cdot D_{\ell} $. Equation~\ref{eq:Sbound2} implies
%
%\begin{eqnarray}
%\label{eq:finalblock1}
%\sum_{\ell=1}^{|J+1|} r_{J(\ell)} D_{\ell} \leq (2\lambda\cdot T +25\lambda^2)\cdot \sum_{\ell=1}^{|J|} r_{J(\ell)} D_{\ell}. 
%\end{eqnarray}
%\vspace{-0.8cm}
\subsubsection*{Putting Everything Together}
Using Corollary~\ref{cor:rjestimator_dlincreasing} we have
$\frac{1}{2\lambda}\cdot D_{\ell} \leq \widehat{R_{I_{sign}(\ell)}} \leq \frac{5}{2}\cdot D_{\ell}$ for all $\ell \in \lbrace 1, \ldots, \vert I_{sign} \vert \rbrace$ which with (\ref{eq:OPTlower1}) and (\ref{eq:OPTupper1}) gives
$\frac{1}{2\lambda\cdot (1+2\lambda\cdot T +25\lambda^2)}\cdot w(M)\leq \sum_{\ell=1}^{|I_{sign}|} r_{I_{sign}(\ell)} \cdot \widehat{R_{I_{sign}(\ell)}} \leq \frac{5}{2}\cdot w(M).$
Recall that we set $T =  8\lambda^2-2\lambda$. Now, folding in the factor of $\frac{1}{8}$ from the partitioning and rescaling the estimator gives an $O(\lambda^4)$-estimation on the weight of an optimal weighted matching.
%Since every edge insertion and deletion supplies the edge weight, it is straightforward to determine the rank for each edge upon every update. For every union $\bigcup_{i=j}^t E_j$ of the $O(\log n)$ ranks we run the black box estimation algorithm, leading to a space increase of $O(\log n)$.
\end{proof}

\setcounter{theorem}{2}
\subsection{Applications}
\label{sec:appl}
Since every edge insertion and deletion supplies the edge weight, it is straightforward to determine the rank for each edge upon every update. Using the following results for unweighted matching, we can obtain estimates with similar approximation guarantee and space bounds for weighted matching.
%\vspace{-0.3cm}
\subsubsection*{Random Order Streams}
For an arbitrary graph whose edges are streamed in random order, Kapralov, Khanna and Sudan~\cite{KKS14} gave an algorithm with $\mathrm{polylog~}n$ approximation guarantee using $\mathrm{polylog~}n$ space with failure probability $\delta = 1/\polylog n$. Since this probability takes the randomness of the input permutation into account, we cannot easily amplify it, though for $\log W \leq \delta$, the extension to weighted matching still succeeds with at least constant probability.
%\vspace{-0.6cm}
\subsubsection*{Adversarial Streams}
The arboricity of a graph $G$ is defined as $\max\limits_{U \subseteq V} \left\lceil \frac{\vert E(U) \vert}{\vert U \vert -1} \right\rceil$. Examples of graphs with constant arboricity include planar graphs and graphs with constant degree. 
For graphs of bounded arboricity $\nu$, Esfandiari et al.~\cite{EHL15} gave an algorithm with an $O(\nu)$ approximation guarantee using $\tilde{O}(\nu \cdot n^{2/3})$ space.
%\vspace{-0.5cm}
\subsubsection*{Dynamic Streams}

We give two estimation algorithms for the size of a maximum matching. First, we see that it is easy to estimate the matching size in trees. Second, we extend the result from \cite{EHL15} where the matching size of so called bounded arboricity graphs in insertion-only streams is estimated to dynamic graph streams.
\paragraph{Matching Size of Trees} Let $T = (V,E)$ be a tree with at least $3$ nodes and let $h_T$ be the number of internal nodes, i.e. nodes with degree greater than $1$. We know that the size of a maximum matching is between $h_T/2$ and $h_T$. Therefore, it suffices to estimate the number of internal nodes of a tree to approximate the maximum matching within $2+\varepsilon$ factor which was also observed in \cite{EHL15}. In order to estimate the matching size, we maintain an $\ell_0$-Estimator for the degree vector $d \in \mathbb{R}^N$ such that $d_v = deg(v)-1$ holds at the end of the stream and with it $\ell_0(d) = h_T$. In other words, we initialize the vector by adding $-1$ to each entry and update the two corresponding entries when we get an edge deletion or insertion. Since the number of edges in a tree is $N-1$, the preprocessing time can be amortized during the stream. Using Theorem 10 from Kane et al.~\cite{KNW10}, we can maintain the $\ell_0$-Estimator for $d$ in $O(\varepsilon^{-2}\log^2 N)$ space.

\begin{theorem}
Let $T = (V,E)$ be a tree with at least $3$ nodes and let $\varepsilon \in (0,1)$. Then there is an algorithm that estimates the size of a maximum matching in $T$ within a $(2+\varepsilon)$-factor in the dynamic streaming model using $1$-pass over the data and $O(\varepsilon^{-2}\log^2 N)$ space.
\end{theorem}
As in \cite{EHL15} this algorithm can be extended to forests with no isolated node.

\paragraph{Matching Size in Graphs with Bounded Arboricity} The algorithm is based on the results from \cite{EHL15}. Since we need parametrized versions of their results, we summarize and rephrase the ideas and proofs in this section. Let $G = (V,E)$ be a graph. The arboricity $a(G)$ of $G$ is a kind of density measure: The number of edges in every induced subgraph of size $s$  in $G$ is bounded by $s \cdot a(G)$. Formally, the arboricity $a(G)$ of $G$ is defined by $a(G) = \max\limits_{U \subseteq V} \left\lceil \frac{\vert E(U) \vert}{\vert U \vert -1} \right\rceil$.
If $\mu_G$ is an upper bound on the average degree of every induced subgraph of $G$ then $\mu_G \leq 2 \cdot a(G)$. 
\begin{definition}[\cite{EHL15}]
A node $v \in V$ is \emph{light} if $deg(v) \leq C$ with $C = \lceil \mu_G \rceil + 3$. Otherwise, $v$ is \emph{heavy}. An edge is \emph{shallow} if and only if both of its endpoints are light. We denote by $h_G$ the number of heavy nodes in $G$ and by $s_G$ the number of shallow edges in $G$, respectively.
\end{definition}
Using the results from Czygrinow, Hanchowiak, and Szymanska \cite{CHS09} (and $C = 20a(G)/\varepsilon^2$) it is possible to get a $O(a(G))$ approximation for the size of a maximum matching by just estimating $h_G$ and $s_G$. Esfandiari et al. \cite{EHL15} improved the approximation factor to roughly $5 \cdot a(g)$.
\begin{lemma}[\cite{EHL15}]
\label{lem:match_arboricity}
Let $G = (V,E)$ be a graph with maximum matching $M^*$. Then we have 
$ \frac{\max{\lbrace h_G, s_G \rbrace}}{\eta} \leq \vert M^* \vert \leq h_G + s_G$
where $\eta = 1.25 C + 0.75$ where $C$ is at most $\lceil 2 a(G) +3 \rceil$.
\end{lemma}
Estimating $h_G$ and $s_G$ is possible by random sampling: For heavy nodes, we randomly draw a large enough set of nodes and count the heavy nodes by maintaining their degree. Rescaling the counter gives a sufficiently good estimate, provide $h_G$ is large enough. For $s_G$ we randomly draw nodes and maintain the induced subgraph. For each node contained in the subgraph it is straightforward to maintain the degree and thereby to decide whether or not a given edge from the subgraph is shallow. Then we can rescale the counted number of shallow edges which gives us an estimation on $s_G$ if $s_G$ is large enough.
Dealing with small values of $s_G$ and $h_G$, Esfandiari et al.  additionally maintain a small maximal matching of size at most $n^\alpha$ with $\alpha < 1$. If the maintained matching exceeds this value then we know that either $s_G$ or $h_G$ is greater than $n^\alpha/2$ by Lemma \ref{lem:match_arboricity} and the estimation of the parameters $h_G$ and $s_G$ will be sufficiently accurate. The main tool to extend this algorithm to dynamic graph streams is to estimate the size of a small matching by means of the Tutte matrix. But first, we restate the following three lemmas from \cite{EHL15} for arbitrary parameters and extend them to dynamic streams.

%In the insertion-only setting \cite{EHL15} the space bottleneck is the estimation of the shallow edges, which needs $\tilde{O}(n^{2/3})$ space for only one pass (which also requires that $\alpha = 2/3$). By allowing a second pass over the data, they showed that only $\tilde{O}(\sqrt{n})$ space is sufficient to estimate $s_G$ (and $\alpha = 1/2$). 
%Adapting the estimation algorithms \ref{alg_estimate_hg} and \ref{alg_estimate_sg} for $h_G$ and $s_G$ for dynamic streams is straightforward. However, maintaining a matching of size at most $n^\alpha$ in one pass is hard when edge deletions are allowed. Nevertheless, we show that a small matching can be maintained in two passes with sublinear (w.r.t. the number of edges) space by utilizing sparse recovery sketches. 

%Our algorithm is similar to the multipass algorithm by Latanzi et al. \cite{LaMSV11} which repeatedly samples $O(n^{1+1/p})$ (free) edges and finds a maximal matching in the sampled subgraph. They showed that after $p$ passes the matching is maximal with high probability. The number of edges in \cite{LaMSV11} has to be $n^{1+c}$ for some $c > 0$. Here, we only have a linear number of edges and we can not apply the same technique to maintain a small matching with sublinear space

%Replacing the subroutines in the estimation algorithm in \cite{fehlt} by Lemma \ref{lem:maintain_matching} and the adaptions of algorithm \ref{alg_estimate_hg} and \ref{alg_estimate_sg} for dynamic streams gives us the following result.

%First, we restate the following three lemmas from \cite{EsfanHLMO2015} for arbitrary parameters.
\begin{lemma}
\label{lem:heavy}
Let $T$ be an integer and $\varepsilon \leq 1/\sqrt{3}$. Then there exists a $1$-pass algorithm for dynamic streams that outputs a value $\widehat{h}$ which is a $(1\pm \varepsilon)$ estimation of $h_G$ if $h_G \geq T$ and which is smaller than $3T$ otherwise. The algorithm needs $O\left(\frac{\log^2 n}{\varepsilon^2} \cdot \frac{n}{T} \right)$ space and succeeds with high probability.
\end{lemma}
\begin{proof}
The probablity of sampling a heavy node is $\frac{h_G}{n}$. Hence, sampling a set of nodes $S$ gives us $|S|\cdot \frac{h_G}{n}$ heavy nodes on expectation. Set $|S|=\frac{3\log n}{\varepsilon^2} \frac{n}{T}$. For each node $v\in S$ we maintain its degree using $O(\log N)$ space. We define the indicator variable $X_v$ with $v \in S$ which is $1$ if $v$ is heavy and $0$ otherwise. Then our estimator for $h_G$ is $\hat{h} = \frac{n}{\vert S \vert} \sum X_v$ which is equal to $h_G$ in expectation. First, assume $h_G \geq T$. Then using the Chernoff bound %(Theorem \ref{thm:chernoff}), 
we have
\begin{eqnarray*}
\Prob{ \hat{h} \geq (1+\varepsilon) \cdot \Ex{\hat{h}}} & = & \Prob{ \sum_{v\in S} X_v \geq (1+\varepsilon) \cdot \Ex{ \sum_{v\in S} X_v}}\\
&\leq & \exp\left( - \frac{3\log n}{\varepsilon^2} \frac{n}{T} \cdot \frac{h_G}{n}\cdot \frac{\varepsilon^2}{3} \right) \leq \frac{1}{n}.
\end{eqnarray*}
The same bound also holds for $\Prob{\hat{h} \leq (1-\varepsilon) \cdot \Ex{\hat{h}}}$. If $h_G < T$, then again using the Chernoff bound gives us
\begin{eqnarray*}
& &\Prob{\frac{n}{|S|}\cdot\left(\sum_{v\in S} X_v\right) \geq 3T} \\
&=& \Prob{\sum_{v\in S} X_v \geq \frac{3 T\cdot |S|\cdot h_G}{n\cdot h_G}}\\
&=& \Prob{\sum_{v\in S} X_v > \left(1 + \frac{3T}{ h_G}-1\right)\cdot \Ex{\sum_{v\in S} X_v}}\\
&\leq & \exp\left(-\frac{3\log n}{\varepsilon^2} \frac{n}{T} \cdot \frac{h_G}{n}\cdot \frac{\frac{3T}{h_G}-1}{2}\right) \\
& \leq & \exp\left(-\frac{3\log n}{\varepsilon^2} \frac{n}{T} \cdot \frac{h_G}{n}\cdot \frac{\frac{2T}{h_G}}{2}\right) \leq \frac{1}{n},
\end{eqnarray*}
where the last inequality follows from $\varepsilon \leq \frac{1}{\sqrt{3}}$.
\end{proof}

\begin{lemma}
\label{lem:shallow2}
Let $T$ be an integer and $\varepsilon \leq 1/\sqrt{3}$. Then there exists a $2$-pass algorithm for dynamic streams that outputs a value $\widehat{s}$ which is a $(1\pm \varepsilon)$ estimation of $s_G$ if $s_G \geq T$ and which is smaller than $3T$ if $s_G < T$.  The algorithm uses $O\left(\frac{a(G) \cdot n \log^4 n}{\varepsilon^2 T} \right)$ space and succeeds with high probability.
\end{lemma}
\begin{proof}
In the first pass, we sample $\frac{3\log n}{\varepsilon^2} \frac{a(G) \cdot n}{T}$ edges uniformly at random using $\ell_0$ samplers, each of which cost at most $O(\log^3 n)$ space~\cite{JST11}. For each node of a sampled edge, we maintain its degree in the second pass to decide whether a given edge is shallow or not. Hereafter, we reapply the analysis of Lemma~\ref{lem:heavy}: Let $S = (e_1, \ldots, e_{\vert S \vert})$ be the sequence of sampled edges in the first pass and let $X_i$ be the indicator variable which is $1$ if and only if $e_i$ is shallow. The probability of sampling a shallow edge is $\frac{s_G}{\vert E \vert}$ which implies that $\Ex{\sum X_i} = \vert S \vert \cdot  \frac{s_G}{\vert E \vert} \geq \vert S \vert \cdot  \frac{s_G}{a(G) \cdot N}$. Now, let $\widehat{s} = \frac{\vert E \vert}{\vert S \vert} \sum X_i$ be our estimator. We know that $\Ex{\widehat{s}} = s_G$. If $s_G \geq T$ then by Chernoff we have
\begin{eqnarray*}
\Prob{ \hat{s} \geq (1+\varepsilon) \cdot \Ex{\hat{s}}} & = & \Prob{ \sum X_i \geq (1+\varepsilon) \cdot \Ex{ \sum X_i}}\\
&\leq & \exp\left( - \frac{3\log n}{\varepsilon^2} \frac{a(G) \cdot n}{T} \cdot \frac{s_G}{a(G) \cdot n} \cdot \frac{\varepsilon^2}{3} \right) \leq \frac{1}{n}.
\end{eqnarray*}
The same bound also holds for $\Prob{\hat{s} \leq (1-\varepsilon) \cdot \Ex{\hat{s}}}$. If $s_G < T$, then again using the Chernoff bound gives us
\begin{eqnarray*}
& &\Prob{\frac{\vert E \vert}{|S|}\cdot\left(\sum X_i \right) \geq 3T} \\
&=& \Prob{\sum X_i \geq \frac{3 T\cdot |S|\cdot s_G}{\vert E \vert \cdot s_G}}\\
&=& \Prob{\sum X_i > \left(1 + \frac{3T}{ s_G}-1\right)\cdot \Ex{\sum X_i}}\\
&\leq & \exp\left(-\frac{3\log n}{\varepsilon^2} \frac{a(G) \cdot n}{T} \cdot \frac{s_G}{a(G) \cdot n}\cdot \frac{\frac{3T}{s_G}-1}{2}\right) \\
& \leq & \exp\left(-\frac{3\log n}{\varepsilon^2} \frac{a(G) \cdot n}{T} \cdot \frac{s_G}{a(G) \cdot n}\cdot \frac{\frac{2T}{s_G}}{2}\right) \leq \frac{1}{n},
\end{eqnarray*}
where the last inequality follows from $\varepsilon \leq \frac{1}{\sqrt{3}}$.
\end{proof}

\begin{lemma}
\label{lem:shallow1}
Let $\varepsilon > 0$ and $T>(16C/\varepsilon)^2$ be an integer. Then there exists a $1$-pass algorithm for dynamic streams that outputs a value $\widehat{s}$ which is a $(1\pm \varepsilon)$ estimation of $s_G$ if $s_G \geq T$ and which is smaller than $3T$ if $s_G < T$.  The algorithm uses $\tilde{O}\left(\frac{a(G) \cdot n}{\varepsilon \sqrt{T}} \right)$ space and succeeds with constant probability.
\end{lemma}
\begin{proof}
Let $S$ be a set of $\frac{4n}{\varepsilon\sqrt{T}}$ randomly chosen nodes. We maintain the entire subgraph induced by $S$ and the degree of each node in $S$. Note that the number of edges in this subgraph at the end of the stream is at most $a(G) \cdot |S|$. Since we have edge deletions this number may be exceeded at some point during the stream. Thus, we cannot explicitly store the subgraph but we can recover all entries using an $a(G) \cdot |S|$-sparse recovery sketch using $\tilde{O}(a(G) \cdot |S|)$ space (see Barkay et al \cite{BaPS13}). Let $e_1, \ldots, e_{s_G}$ be the shallow edges in $G$. Define $X_i=1$ if $e_i \in E(S)$ and $0$ otherwise. $X_i$ is Bernouilli distributed where the probability of both nodes being included in the subgraph follows from the hypergeometric distribution with population $n$, $2$ successes in the population, sample size $|S|$ and $2$ successes in the sample:
\[p= \frac{\binom{2}{2}\binom{N-2}{|S|-2}}{\binom{n}{|S|}} = \frac{|S|\cdot (|S|-1)}{n\cdot (n-1)} \geq \dfrac{\vert S \vert^2}{2 n^2} = \dfrac{8}{\varepsilon^2 T}.\]
Hence $X_i$ is Bernoulli distributed, we have $Var\left[X_i\right]= p\cdot (1-p)\leq p$. We know that $Var\left[\sum X_i\right] = \sum Var\left[X_i\right] + \sum_{i \neq j} \Cov{X_i, X_j}$. For the covariance between two variables $X_i$ and $X_j$ we have two cases: If $e_i$ and $e_j$ do not share a node, then $X_i$ and $X_j$ cannot be positively correlated, i.e. $\Cov{X_i, X_j} > 0$. To be more precise, we observe that by definition $\Cov{X_i,X_j}$ is equal to $\Ex{X_i X_j} - \Ex{X_i} \cdot \Ex{X_j}$ which is equal to $\Prob{X_i = X_j = 1} - p^2$. The probability $\Prob{X_i = X_j = 1}$ is equal to the probability of drawing exactly $4$ fixed nodes from $V$ with a sample of size $\vert S \vert$ which is
$$ \frac{\binom{4}{4}\binom{n-4}{|S|-4}}{\binom{n}{|S|}} = \frac{|S|\cdot (|S|-1)\cdot (|S|-2)\cdot (|S|-3)}{n\cdot (n-1)\cdot (n-2)\cdot (n-3)}.$$
Since $\frac{a+c}{b+c} \geq \frac{a}{b}$ for $a \leq b$ and $c \geq 0$, this probability is at most $p^2$ which means that the covariance is at most $0$. If $e_i$ and $e_j$ share a node, we have
\begin{eqnarray*}
\Cov{X_i,X_j} &\leq&  \Prob{X_i=X_j=1}\\
&=&\frac{\binom{3}{3}\binom{n-3}{|S|-3}}{\binom{n}{|S|}} = \frac{|S|\cdot (|S|-1)\cdot (|S|-2)}{n\cdot (n-1)\cdot (n-2)} \leq p^{3/2}.
\end{eqnarray*}
By definition each node incident to a shallow edge has at most $C$ neighbors and therefore, we have at most $2C$ edges that share a node with a given shallow edge. In total, we can bound the variance of $X$
\begin{eqnarray*}
\Var{X} &=& \sum \Var{X_i} + \sum_{i \neq j} \Cov{X_i,X_j} \\
&\leq & p \cdot s_G + \sum_{\substack{
   e_i \neq e_j,\\
   e_i,e_j \text{ share a node}
  }} \Cov{X_i,X_j} \leq p \cdot s_G  + 2C \cdot s_G \cdot p^{3/2} \leq 2p \cdot s_G
\end{eqnarray*} 
where the last inequality follows from $\sqrt{p} \leq \frac{\vert S \vert}{N/2} = \frac{8}{\varepsilon \sqrt{T}}$ and $T \geq (16C/\varepsilon)^2$. Using Chebyshev's inequality we have for $s_G\geq T$
\begin{eqnarray*}
\Prob{\left\vert\frac{1}{p} \cdot X - \frac{1}{p}\Ex{X}\right\vert > \epsilon\cdot \frac{1}{p}\Ex{X}}& = & \Prob{ \left\vert X - \Ex{X}\right\vert > \epsilon\cdot \Ex{X}}\\
&\leq & \frac{\Var{X}}{\epsilon^2\Ex{X}^2} \leq \frac{2p \cdot s_G}{\epsilon^2 p^2 \cdot s_G^2} \leq \frac{2}{\epsilon^2 T p}\\
&\leq & \frac{2\varepsilon^2 T}{8 \epsilon^2 T} = \frac{1}{4}.
\end{eqnarray*}

If $s_G < T$, we have $\Ex{X} = p \cdot s_G < pT$. Thus, it is 
\begin{eqnarray*}
\Prob{\frac{1}{p} \cdot X \geq 3T} &=& \Prob{X - \Ex{X} \geq 3 T p - \Ex{X}}\\
& \leq & \Prob{\vert X - \Ex{X} \vert \geq 2 T p} \\
&\leq & \frac{ \Var{X}}{4 T^2 p^2} \leq \frac{2 p \cdot s_G}{4 T^2 p^2} \leq \frac{2}{4 T p} \leq \frac{2\varepsilon^2 T}{16T} = \dfrac{\varepsilon^2}{16} \leq \frac{1}{16}.
\end{eqnarray*}
\end{proof}

\begin{algorithm}[h]
\caption{\bf{Unweighted Matching Approximation}}
\label{alg:unweighted_matching1}
\algorithmicrequire{ $G=(V,E)$ with $a(G) \leq \alpha$ and $\varepsilon \in (0, 1/\sqrt{3})$}\\
\algorithmicensure{ Estimator on the size of a maximum matching}
\begin{algorithmic}
\State{Set $T = n^{2/5}$ for a single pass and $T=n^{1/3}$ for two passes and $\eta = 2.5\lceil 2 \cdot \alpha+3 \rceil+5.75$.}
\State{Let $\hat{h}$ and $\hat{s}$ be the estimators from Lemma \ref{lem:heavy} and Lemma \ref{lem:shallow1}}
\For{$i=0, \ldots, \log 3T/(1-\varepsilon)$}
\State{Solve rank decision with parameter $k = 2^i$ on the Tutte-Matrix $T(G)$ with randomly chosen indeterminates}
\EndFor
\If{$\text{rank}(T(G)) < 3T/(1-\varepsilon)$}
\State{Output the maximal $2^{i+1}$ for the maximal $i \in \lbrace 0,\ldots , 2^{\log 3T/(1-\varepsilon)} \rbrace$ with $\text{rank}(T(G)) \geq 2^i$}
\Else
\State{Output $\dfrac{\max\lbrace \hat{h}, \hat{s} \rbrace}{(1+\varepsilon) \eta}$.}
\EndIf
\end{algorithmic}
\end{algorithm}

Algorithm \ref{alg:unweighted_matching1} shows the idea of the estimation of the unweighted maximum matching size in bounded arboricity graphs using the previous results and the relation between the rank of the Tutte matrix and the matching size. 
\begin{theorem}
\label{thm:unweighted}
Let $G$ be a graph with $a(G) \leq \alpha$ with $n \geq (16\alpha/\varepsilon)^{5}$. Let $\varepsilon \in (0,1/\sqrt{3})$. Then there exists an algorithm estimating the size of the maximum matching in $G$ within a $\frac{2(1+\varepsilon)(5\cdot a(G)+O(1))}{(1-\varepsilon)}$-factor in the dynamic streaming model using 
\begin{itemize}
\item a single pass over the data and $\tilde{O}(\frac{\alpha\cdot n^{4/5}}{\varepsilon^2})$ space or
\item $2$ passes over the data and $\tilde{O}(\alpha \cdot n^{2/3})$ space.
\end{itemize}
\end{theorem}
\begin{proof}
For the sake of simplicity we assume that $3T/(1-\varepsilon)$ is a power of two. We know that we can decide the rank decision problem with parameter $k$ in a dynamic stream with one pass using $O(k^2 \log n)$ space by Theorem 5.1 of Clarkson and Woodruff \cite{CW09}. Thus, invoking this algorithm for $k = 2^0, 2^1, \ldots, 2^{\log 3T/(1-\varepsilon)}$ results in a space requirement of $O(T^2 \cdot \log T \cdot \log n) = O(T^2 \log^2 N)$ for our choices of $T$.
%The first part of the theorem follows from Corollary~\ref{cor:2pass} and Lemmas~\ref{lem:heavy} and~\ref{lem:shallow2} with $T=\sqrt{n}$.
%For the first part of the theorem we assume that in Algorithm \ref{alg:unweighted_matching} each indeterminate of the Tutte-matrix independently gets a weight uniformly chosen from $\{1, \ldots, \text{poly}(n)\}$. Upon insertion or deletion of an edge we update the corresponding entry by the appropriate value. For all powers of $2$ less or equal than $18T$, we solve the rank decision problem on the Tutte-matrix of the graph where the indeterminates have been replaced by the random values using a total of $O(T^2\log^2 n\log T ) = O(T^2\log^3 n)$ space. 
For the first part of the theorem, we estimate $s_G$ and $h_G$ in $1$-pass by $\hat{h}$ and $\hat{s}$ using $\tilde{O}\left(\frac{n}{\varepsilon^2 T}\right)$ and $\tilde{O}\left(\frac{\alpha\cdot n}{\varepsilon \sqrt{T}}\right)$ space, see Lemma \ref{lem:heavy} and Lemma~\ref{lem:shallow1}. Setting $T=n^{2/5}$ gives us the desired space bound of $\tilde{O}\left(\frac{\alpha\cdot n^{4/5}}{\varepsilon^2}\right)$ (note that $T > (16\alpha/\varepsilon)^{2}$ which is required for Lemma \ref{lem:shallow1}). For the second part of the theorem, we can improve the space requirements for the estimator $\hat{h}$ and $\hat{s}$ to $\tilde{O}\left(\frac{a(G) n}{T}\right)$ by using Lemma \ref{lem:heavy} and Lemma \ref{lem:shallow2}. Now, setting $T = n^{1/3}$ gives the desired space bound.

Let $OPT$ be the size of a maximum matching. First, we check whether $OPT \geq 2 \cdot 3T/(1-\varepsilon)$ by invoking the rank decision algorithm with parameter $k = 3T/(1-\varepsilon)$. Since the rank of the matrix is equal to $2OPT$, this decides whether $OPT \geq  2 \cdot 3T/(1-\varepsilon)$. If this is not true, we can give a $2$-approximation on $OPT$ by testing whether the rank of the Tutte matrix is in $[2^i, 2^{i+1})$ for $i = 0, \ldots, \log{(3T/(1-\varepsilon))}-1$. If $OPT \geq 2 \cdot  3T/(1-\varepsilon)$ Lemma \ref{lem:match_arboricity} implies that $\max\lbrace h_G, s_G\rbrace \geq 3T/(1-\varepsilon)$ since $h_G+s_G \geq OPT$. Assuming that we can approximate $\max\lbrace h_G, s_G\rbrace$ then again by Lemma \ref{lem:match_arboricity} we can estimate $OPT$ since 
$$ \dfrac{\max\lbrace h_G, s_G \rbrace}{\eta} \leq OPT \leq h_G+s_G \leq 2\max\lbrace h_G, s_G \rbrace. $$
W.l.o.g. let $\widehat{h} = \arg\max\lbrace \hat{h}, \hat{s} \rbrace$. Now we have two cases:
\begin{enumerate}
\item If $h_G = \arg\max\lbrace h_G, s_G \rbrace \geq T$ then by Lemma \ref{lem:heavy} $\widehat{h}$ is a $(1\pm\varepsilon)$ estimation on $h_G$.
\item If $s_G = \arg\max\lbrace h_G, s_G \rbrace \geq 3T/(1-\varepsilon)$ we know by Lemma \ref{lem:shallow1} that $\widehat{s} \geq 3T$ which implies that $\widehat{h} \geq \widehat{s} \geq 3T$. Thus by Lemma \ref{lem:heavy} $\widehat{h}$ is a $(1\pm\varepsilon)$ estimation on $h_G$. This gives us
$$ (1-\varepsilon) s_G \leq \widehat{s} \leq \widehat{h} \leq (1+\varepsilon)h_G \leq (1+\varepsilon) s_G. $$
\end{enumerate}
Therefore, $\max\lbrace \hat{h}, \hat{s} \rbrace$ is a good estimator for $\max\lbrace h_G, s_G \rbrace$. For the estimator $\frac{\max\lbrace \hat{h}, \hat{s} \rbrace}{(1+\varepsilon) \eta}$ we have
$$ \frac{(1-\varepsilon)}{2(1+\varepsilon)\eta} \cdot OPT \leq \frac{(1-\varepsilon)\max\lbrace h_G, s_G \rbrace}{(1+\varepsilon) \eta} \leq \frac{\max\lbrace \hat{h}, \hat{s} \rbrace}{(1+\varepsilon) \eta} \leq \frac{(1+\varepsilon)\max\lbrace h_G, s_G \rbrace}{(1+\varepsilon) \eta} \leq OPT.$$
%
%If $\max\lbrace \hat{h}, \hat{s} \rbrace > 4T$ we know that $\max\lbrace \hat{h}, \hat{s} \rbrace$ approximates the value of $\max\lbrace h_G, s_G \rbrace$ and output $\dfrac{\max\lbrace \hat{h}, \hat{s} \rbrace}{(1+\varepsilon) \eta}$ as an estimation for the matching size. Using Lemma \ref{lem:match_arboricity} it holds
%$$ \dfrac{\max\lbrace \hat{h}, \hat{s} \rbrace}{(1+\varepsilon) \eta} \leq M^* \leq h_G+s_G \leq 2(1+\varepsilon)\max\lbrace \hat{h}, \hat{s} \rbrace. $$
%If $\max\lbrace \hat{h}, \hat{s} \rbrace < 4T$ we know that $M^* \leq 9 T$. Thus, we can use the rank decisions to estimate the size up to a factor of $2$.
%Besides the random bits for the Tutte-matrix the algorithm uses $\tilde{O}(\nu \cdot n^{4/5})$ for the single pass estimation of $s_G$. Using Nisan's pseudo random generators for bounded space computation~\cite{Nis92,Indyk00} the random bits can be drawn from a seed of length $O(S \log n)$ (where $S$ is the space requirement of our algorithm).
\end{proof}
\section{Lower Bound}
\label{sec:lower}
Esfandiari et al. \cite{EHL15} showed a space lower bound of $\Omega(\sqrt{n})$ for any estimation better than $3/2$. Their reduction (see below) uses the Boolean Hidden Matching Problem introduced by Bar-Yossef et al.~\cite{BJK08}, and further studied by Gavinsky et al.~\cite{GKKRW08}. We will use the following generalization due to Verbin and Yu \cite{VerW11}.
%\vspace{-0.1cm}
\begin{definition}[Boolean Hidden Hypermatching Problem \cite{VerW11}]
In the \emph{Boolean Hidden Hypermatching} Problem $BHH_{t,n}$ Alice gets a vector $x \in \lbrace 0,1 \rbrace^{n}$ with $n = 2kt$ and $k \in \mathbb{N}$ and Bob gets a perfect $t$-hypermatching $M$ on the $n$ coordinates of $x$, \ie each edge has exactly $t$ coordinates, and a string $w \in \lbrace 0,1 \rbrace^{n/t}$. We denote the vector of length $n/t$ given by $(\bigoplus_{1 \leq i \leq t} x_{M_{1,i}}, \ldots, \bigoplus_{1 \leq i \leq t} x_{M_{n/t,i}})$ by $Mx$ where $(M_{1,1}, \ldots, M_{1,t}), \ldots, (M_{n/t,1}, \ldots,  M_{n/t,t})$ are the edges of $M$. 
The problem is to return $1$ if $Mx \oplus w = 1^{n/t}$ and $0$ if $Mx \oplus w = 0^{n/t}$, otherwise the algorithm may answer arbitrarily.
%It is promised that either $Mx \oplus w = 1^{n/t}$ or $Mx \oplus w = 0^{n/t}$. The problem is to return $1$ in the first case and $0$ otherwise.
\end{definition}
%\vspace{-0.1cm}
Verbin and Yu \cite{VerW11} showed a lower bound of $\Omega(n^{1-1/t})$ for the randomized one-way communication complexity for $BHH_{t,n}$. For our reduction we require $w=0^{n/t}$ and thus $Mx = 1^{n/t}$ or $Mx = 0^{n/t}$. We denote this problem by $BHH^0_{t,n}$. We can show that this does not reduce the communication complexity.
%\vspace{-0.2cm}
%\begin{definition}
%The problem $BHH^0_{t,n}$ is the same as the $BHH_{t,n}$ problem with $w$ fixed to be $0^{n/t}$ and $x \in \lbrace 0,1 \rbrace^n$ has exactly $n/2$ bits equal to $1$.
%\end{definition}
%\vspace{-0.4cm}
\begin{lemma}
The communication complexity of $BHH^0_{t,4n}$ is lower bounded by the communication complexity of $BHH_{t,n}$.
\end{lemma}
\begin{proof}
First, let assume that $t$ is odd. Let $x \in \lbrace 0,1 \rbrace^{n}$ with $n = 2kt$ for some $k \in \mathbb{N}$ and $M$ be a perfect $t$-hypermatching on the $n$ coordinates of $x$ and $w \in \lbrace 0,1 \rbrace^{n/t}$. We define $x' = [x^T x^T \overline{x}^T \overline{x}^T]^T$  to be the concatenation of two identical copies of x and two identical copies of the vector resulting from the bitwise negation of $x$. W.l.o.g. let $\lbrace x_1, \ldots, x_t \rbrace \in M$ be the $l$-th hyperedge of $M$. Then we add the following four hyperedges to $M'$:
\begin{itemize}
\item $\lbrace x_1, \overline{x_2}, \ldots, \overline{x_t} \rbrace, \lbrace \overline{x_1}, x_2, \overline{x_3}, \ldots, \overline{x_t} \rbrace$, $\lbrace \overline{x_1}, \overline{x_2}, x_3, \ldots, x_t \rbrace$, and $\lbrace x_1, \ldots, x_t \rbrace$ if $w_l = 0$,
\item $\lbrace \overline{x_1}, x_2, \ldots, x_t \rbrace, \lbrace x_1, \overline{x_2}, \ldots, x_t \rbrace$, $\lbrace x_1, x_2, \overline{x_3}, \ldots, \overline{x_t} \rbrace$, and $\lbrace \overline{x_1}, \ldots, \overline{x_t} \rbrace$ if $w_l = 1$.
\end{itemize}
The important observation here is that we flip even number of bits in the case $w_l = 0$ and an odd number of bits if $w_l = 1$ (since $t$ is odd).
Since every bit flip results in a change of the parity of the set of bits, the parity does not change if we flip an even number of bits and the parity also flips if we negate an odd number of bits. Therefore, if $w_l$ is the correct (respectively wrong) parity of $\lbrace x_1, \ldots, x_t \rbrace$ then the parity of the added sets is $0$ (respectively $1$), \ie $M'x' = 0^{2n}$ if $Mx \oplus w = 0^{n/2}$ and $M'x' = 1^{2n}$ if $Mx \oplus w = 1^{n/2}$. The number of ones in $x' \in \lbrace 0,1 \rbrace^{4n}$ is exactly $2n$. If $t$ is even, we can just change the cases for the added edges such that we flip an even number of bits in the case $w_l = 0$ and an odd number of bits if $w_l = 1$. Overall, this shows that a lower bound for $BHH_{t,n}$ implies a lower bound for $BHH^0_{t,4n}$.
\end{proof}
%\vspace{-0.2cm}
%Let us now sketch the reduction from $BHH_{2,n}^0$ to approximate maximum matching to get the idea how to extend it to the general bound. Let $x,M$ be the input for Alice and Bob. They construct a graph consisting of $2n$ nodes denoted by $v_{1,i}$ and $v_{2,i}$, for $i\in\{ 1,\ldots , n\}$. For each bit $x_i$ of $x \in \lbrace 0,1 \rbrace^n$, Alice adds an edge $\lbrace v_{1,i}, v_{2,i} \rbrace$ iff $x_i = 1$ and sends a message to Bob. Bob adds an edge between $v_{2,i}$ and $v_{2,j}$ for each edge $\lbrace x_i, x_j \rbrace \in M$ and approximates the size of the matching. If all parities are $1$ then the size of the maximum matching is $n/2$. If the parities are all $0$ then the size is $3n/4$. Every streaming algorithm that approximates better than $3/2$ can distinguish between these two cases. The first observation is that the size of the matching is lower bounded by the number of ones in $x$. The second observation is that the added edges by Bob increase the matching iff the parities of all pairs are $0$ and only the edges between the two $0$ input bits of Alice increase the matching. Since it is promised that all parities are equal and the number of ones is exactly $n/2$ we can calculate the number of $(0,0)$ pairs. For our lower bound we show that this calculation is still possible if Bob adds a $t$-clique between the corresponding nodes of the hyperedge.

\begin{figure}[h]
\centering
\begin{tikzpicture}[scale = 0.9]
\node[circle, draw] (a1) at (0,0) {$v_{1,1~}$};
\node[circle, draw] (a2) at (1.5,0) {$v_{1,2~}$};
\node[circle, draw] (a3) at (3,0) {$v_{1,3~}$};
\node[circle, draw] (a4) at (4.5,0) {$v_{1,4~}$};
\node[circle, draw] (a5) at (6,0) {$v_{1,5~}$};
\node[circle, draw] (a6) at (7.5,0) {$v_{1,6~}$};
\node[circle, draw] (a7) at (9,0) {$v_{1,7~}$};
\node[circle, draw] (a8) at (10.5,0) {$v_{1,8~}$};
\node[circle, draw] (a9) at (12,0) {$v_{1,9~}$};
\node[circle, draw] (a10) at (13.5,0) {$v_{1,10}$};
\node[circle, draw] (a11) at (15,0) {$v_{1,11}$};
\node[circle, draw] (a12) at (16.5,0) {$v_{1,12}$};

\node[circle, draw] (b1) at (0,-2) {$v_{2,1~}$};
\node[circle, draw] (b2) at (1.5,-2) {$v_{2,2~}$};
\node[circle, draw] (b3) at (3,-2) {$v_{2,3~}$};
\node[circle, draw] (b4) at (4.5,-2) {$v_{2,4~}$};
\node[circle, draw] (b5) at (6,-2) {$v_{2,5~}$};
\node[circle, draw] (b6) at (7.5,-2) {$v_{2,6~}$};
\node[circle, draw] (b7) at (9,-2) {$v_{2,7~}$};
\node[circle, draw] (b8) at (10.5,-2) {$v_{2,8~}$};
\node[circle, draw] (b9) at (12,-2) {$v_{2,9~}$};
\node[circle, draw] (b10) at (13.5,-2) {$v_{2,10}$};
\node[circle, draw] (b11) at (15,-2) {$v_{2,11}$};
\node[circle, draw] (b12) at (16.5,-2) {$v_{2,12}$};

%\draw (a1) -- (b1);
\draw (a4) -- (b4);
\draw (a6) -- (b6);
\draw (a7) -- (b7);
\draw (a9) -- (b9);
\draw (a11) -- (b11);
\draw (a12) -- (b12);
\draw(b1) -- (b2);
\draw(b3) -- (b2);
\draw (b1) to[out=-45,in=-135] (b3);

\draw (b5) -- (b4);
\draw (b4) to[out=-45,in=-135] (b7);
\draw (b5) to[out=-45,in=-135] (b7);

\draw (b6) to[out=-45,in=-135] (b10);
\draw (b6) to[out=-45,in=-135] (b11);
\draw (b10)--(b11);

\draw (b8) -- (b9);
\draw (b8) to[out=-45,in=-135] (b12);
\draw (b9) to[out=-45,in=-135] (b12);

\end{tikzpicture}
\caption{Worst case instance for $t=3$. Bob's hypermatching corresponds to disjoint $3$-cliques among the lower nodes and Alice' input vector corresponds to the edges between upper and lower nodes.}\label{fig:lowerboundt}
\end{figure}
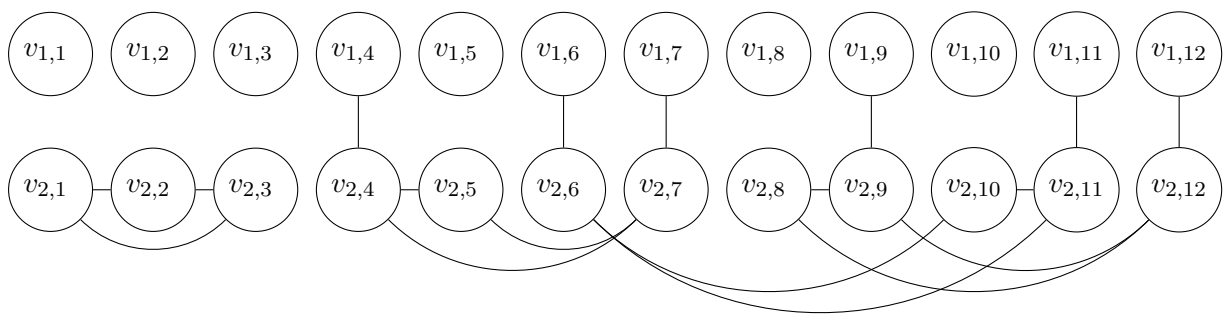

\setcounter{theorem}{1}
\begin{theorem}
Any randomized streaming algorithm that approximates the  maximum matching size within a $1+\frac{1}{3t/2-1}$ factor for $t \geq 2$ needs $\Omega(n^{1-1/t})$ space.
\end{theorem}
\begin{proof}
Let $x,M$ be the input to the $BHH^0_{t,n}$ problem, \ie $M$ is a perfect $t$-hypermatching on the coordinates of $x$, $x$ has exactly $n/2$ ones and it is promised that either $Mx = 0^{n/t}$ or $Mx = 1^{n/t}$. We construct the graph for the reduction as described above: For each bit $x_i$ we have two nodes $v_{1,i}, v_{2,i}$ and Alice adds the edge $\lbrace v_{1,i}, v_{2,i} \rbrace$ iff $x_i = 1$. For each edge $\lbrace x_{i_1}, \ldots, x_{i_t} \rbrace \in M$ Bob adds a $t$-clique consisting of the nodes $v_{2,i_1}, \ldots, v_{2,i_t}$. For now, let us assume $t$ to be odd. We know that the matching is at least $n/2$ because $x$ has exactly $n/2$ ones. Since Bob adds a clique for every edge it is always possible to match all (or all but one) nodes of the clique whose corresponding bit is $0$. In the case of $Mx = 0^{n/t}$ the parity of every edge is $0$, \ie the number of nodes whose corresponding bit is $1$ is even. Let $M_{2i} \subseteq M$ be the hyperedges containing exactly $2i$ one bits and define $l_{2i} := \vert M_{2i} \vert$. Then we know $n/2 = \sum_{i=0}^{\lfloor t/2 \rfloor} 2i \cdot l_{2i}$ and $\vert M \vert = n/t = \sum_{i=0}^{\lfloor t/2 \rfloor} l_{2i}$. For every edge in $M_{2i}$ the size of the maximum matching within the corresponding subgraph is exactly $2i + \lfloor (t-2i)/2 \rfloor = 2i + \lfloor t/2 \rfloor - i$ for every $i = 0,\ldots,\lfloor t/2 \rfloor$ (see Fig. \ref{fig:lowerboundt}). Thus, we have a matching of size
$$ \sum_{i=0}^{\lfloor t/2 \rfloor} (2i+(\lfloor t/2 \rfloor - i)) l_{2i} =\dfrac{n}{2}+\dfrac{t-1}{2} \cdot \dfrac{n}{t} - \dfrac{n}{4} = \dfrac{3n}{4} - \dfrac{n}{2t}. $$
If we have $Mx = 1^{n/t}$ then let $M_{2i+1} \subseteq M$ be the hyperedges containing exactly $2i+1$ one bits and define $l_{2i+1} := \vert M_{2i+1} \vert$. Again, we know $n/2 = \sum_{i=0}^{\lfloor t/2 \rfloor} (2i+1) \cdot l_{2i+1}$ and $\vert M \vert = n/t = \sum_{i=0}^{\lfloor t/2 \rfloor} l_{2i+1}$. For every edge in $M_{2i+1}$ the size of the maximum matching within the corresponding subgraph is exactly $2i+1 + (t-2i-1)/2 = 2i+1 + \lfloor t/2 \rfloor -i$ for every $i = 0,\ldots,\lfloor t/2 \rfloor$. Thus, the maximum matching has a size 
$$ \sum_{i=0}^{\lfloor t/2 \rfloor} (2i+1+(\lfloor t/2 \rfloor - i)) l_{2i+1} = \dfrac{n}{2}+\dfrac{t-1}{2} \cdot \dfrac{n}{t}- \dfrac{1}{2} \sum_{i=0}^{\lfloor t/2 \rfloor} (2i+1) \cdot l_{2i+1}+\dfrac{n}{2t} = \dfrac{3n}{4}. $$

For $t$ even, the size of the matching is 
$$ \sum_{i=0}^{t/2} (2i+(t-2i)/2) l_{2i} = \dfrac{n}{2}+\dfrac{t}{2} \cdot \dfrac{n}{t} - \dfrac{n}{4} = \dfrac{3n}{4} $$
if $Mx = 0^{n/t}$. Otherwise, we have
\begin{eqnarray*}
\sum_{i=0}^{t/2} \left(2i+1+\left\lfloor \dfrac{t-2i-1}{2} \right\rfloor \right) l_{2i+1} &= & \dfrac{n}{2}+\sum_{i=0}^{t/2}(t/2-i-1) l_{2i+1} \\
& =& \dfrac{n}{2}-(t/2-1) \cdot \dfrac{n}{t}-\dfrac{n}{4}+\dfrac{n}{2t} = \dfrac{3n}{4}-\dfrac{n}{2t}.
\end{eqnarray*}

As a consequence, every streaming algorithm that computes an $\alpha$-approximation on the size of a maximum matching with 
$$ \alpha < \dfrac{(3/4)n}{((3/4)-1/(2t))n} = 1/(1-4/6t) = 1+\dfrac{1}{3t/2-1} $$
can distinguish between $Mx = 0^{n/t}$ and $Mx = 1^{n/t}$ and, thus, needs $\Omega(n^{1-1/t})$ space. 
\end{proof}
%\vspace{-0.2cm}
Finally, constructing the Tutte-matrix with randomly chosen entries gives us
%\vspace{-0.1cm}
\begin{corollary}
Any randomized streaming algorithm that approximates $\textup{rank}(A)$ of $A\in\mathbb{R}^{n\times n}$  within a $1+\frac{1}{3t/2-1}$ factor for $t \geq 2$ requires $\Omega(n^{1-1/t})$ space.
\end{corollary}

%\vspace{-0.7cm}
%\section{Conclusion}
%\label{conclusion}
%\vspace{-0.3cm}
%In this paper, we presented an estimation algorithm for weighted matching implementable in any streaming model, as well as estimation algorithms for unweighted matching for certain graph classes in dynamic streams. To our knowledge both problems have not been previously investigated in literature. In addition, we also give a lower bound for small approximation factors for unweighted matching in adversarial streams. Lower bounds for approximation factors greater than $2$ have not been previously investigated and are an interesting open problem, as communication protocols likely have to make use of multiple players. 
%
%~

%Esfandiari et al. \cite{fehlt} investigated a similar problem, which they called Matching Parity Problem, where the promise is that there is a $\theta \in \lbrace 0,1 \rbrace$ such that for at least $2 \beta n$ edges ($\beta \in (1/2,1]$) in $M$ the corresponding values in $Mx$ is $\theta$, and Bob's goal is to compute this $\theta$. They gave a reduction which shows that it is possible to compute $\theta$ with $S+O(\log n)$ communication if there is a streaming algorithm using space $Q$ which returns an $3/2-\varepsilon$-approximation of the size of a maximum matching.

%\vspace{-0.5cm}

\bibliography{literature-short}{}
\bibliographystyle{plain}

\end{document}